\documentclass[lettersize,onecolumn,journal]{IEEEtran}
\usepackage{amsmath,amsfonts}
\usepackage{algorithmic}
\usepackage{algorithm}
\usepackage{array}
\usepackage[caption=false,font=normalsize,labelfont=sf,textfont=sf]{subfig}
\usepackage{textcomp}
\usepackage{stfloats}
\usepackage{url}
\usepackage{verbatim}
\usepackage{graphicx}
\usepackage{cite}
\usepackage{hyperref}
 \hypersetup{colorlinks,linkcolor={blue},citecolor={red},urlcolor={blue}}
\usepackage{tikz,float,multirow,booktabs,enumerate}
\usetikzlibrary{arrows, decorations.markings}
\tikzstyle{vecArrow} = [thick, decoration={markings,mark=at position
	1 with {\arrow[semithick]{open triangle 60}}},
double distance=1.4pt, shorten >= 5.5pt,
preaction = {decorate},
postaction = {draw,line width=1.4pt, white,shorten >= 4.5pt}]
\tikzstyle{innerWhite} = [semithick, white,line width=1.4pt, shorten >= 4.5pt]

\newtheorem{definition}{Definition}
\newtheorem{proposition}[definition]{Proposition}
\newtheorem{lemma}[definition]{Lemma}

\newtheorem{theorem}[definition]{Theorem}
\newtheorem{corollary}[definition]{Corollary}
\newtheorem{conjecture}[definition]{Conjecture}

\newtheorem{remark}[definition]{Remark}
\newtheorem{example}[definition]{Example}
\newtheorem{question}[definition]{Question}

\def\bcj{\begin{conjecture}}
	\def\ecj{\end{conjecture}}
\def\bcr{\begin{corollary}}
	\def\ecr{\end{corollary}}
\def\bd{\begin{definition}}
	\def\ed{\end{definition}}
\def\bea{\begin{eqnarray}}
\def\eea{\end{eqnarray}}
\def\bem{\begin{enumerate}}
	\def\eem{\end{enumerate}}
\def\bex{\begin{example}}
	\def\eex{\end{example}}
\def\bim{\begin{itemize}}
	\def\eim{\end{itemize}}
\def\bl{\begin{lemma}}
	\def\el{\end{lemma}}
\def\bma{\begin{bmatrix}}
	\def\ema{\end{bmatrix}}
\def\bpf{\begin{proof}}
	\def\epf{\end{proof}}
\def\bpp{\begin{proposition}}
	\def\epp{\end{proposition}}
\def\bqu{\begin{question}}
	\def\equ{\end{question}}
\def\br{\begin{remark}}
	\def\er{\end{remark}}
\def\bt{\begin{theorem}}
	\def\et{\end{theorem}}

\def\squareforqed{\hbox{\rlap{$\sqcap$}$\sqcup$}}
\def\qed{\ifmmode\squareforqed\else{\unskip\nobreak\hfil
		\penalty50\hskip1em\null\nobreak\hfil\squareforqed
		\parfillskip=0pt\finalhyphendemerits=0\endgraf}\fi}
\def\endenv{\ifmmode\;\else{\unskip\nobreak\hfil
		\penalty50\hskip1em\null\nobreak\hfil\;
		\parfillskip=0pt\finalhyphendemerits=0\endgraf}\fi}
\newenvironment{proof}{\noindent \textbf{{Proof.~} }}{\qed}
\def\Dbar{\leavevmode\lower.6ex\hbox to 0pt
	{\hskip-.23ex\accent"16\hss}D}
\makeatletter
\def\url@leostyle{%
	\@ifundefined{selectfont}{\def\UrlFont{\sf}}{\def\UrlFont{\small\ttfamily}}}
\makeatother
\def\bcj{\begin{conjecture}}
	\def\ecj{\end{conjecture}}
\def\bcr{\begin{corollary}}
	\def\ecr{\end{corollary}}
\def\bd{\begin{definition}}
	\def\ed{\end{definition}}
\def\bea{\begin{eqnarray}}
\def\eea{\end{eqnarray}}
\def\bem{\begin{enumerate}}
	\def\eem{\end{enumerate}}
\def\bex{\begin{example}}
	\def\eex{\end{example}}
\def\bim{\begin{itemize}}
	\def\eim{\end{itemize}}
\def\bl{\begin{lemma}}
	\def\el{\end{lemma}}
\def\bpf{\begin{proof}}
	\def\epf{\end{proof}}
\def\bpp{\begin{proposition}}
	\def\epp{\end{proposition}}
\def\bqu{\begin{question}}
	\def\equ{\end{question}}
\def\br{\begin{remark}}
	\def\er{\end{remark}}
\def\bt{\begin{theorem}}
	\def\et{\end{theorem}}

\def\btb{\begin{tabular}}
	\def\etb{\end{tabular}}

\newcommand{\nc}{\newcommand}

\nc{\bbA}{\mathbb{A}} \nc{\bbB}{\mathbb{B}} \nc{\bbC}{\mathbb{C}}
\nc{\bbD}{\mathbb{D}} \nc{\bbE}{\mathbb{E}} \nc{\bbF}{\mathbb{F}}
\nc{\bbG}{\mathbb{G}} \nc{\bbH}{\mathbb{H}} \nc{\bbI}{\mathbb{I}}
\nc{\bbJ}{\mathbb{J}} \nc{\bbK}{\mathbb{K}} \nc{\bbL}{\mathbb{L}}
\nc{\bbM}{\mathbb{M}} \nc{\bbN}{\mathbb{N}} \nc{\bbO}{\mathbb{O}}
\nc{\bbP}{\mathbb{P}} \nc{\bbQ}{\mathbb{Q}} \nc{\bbR}{\mathbb{R}}
\nc{\bbS}{\mathbb{S}} \nc{\bbT}{\mathbb{T}} \nc{\bbU}{\mathbb{U}}
\nc{\bbV}{\mathbb{V}} \nc{\bbW}{\mathbb{W}} \nc{\bbX}{\mathbb{X}}
\nc{\bbZ}{\mathbb{Z}}

\nc{\bA}{{\bf A}} \nc{\bB}{{\bf B}} \nc{\bC}{{\bf C}}
\nc{\bD}{{\bf D}} \nc{\bE}{{\bf E}} \nc{\bF}{{\bf F}}
\nc{\bG}{{\bf G}} \nc{\bH}{{\bf H}} \nc{\bI}{{\bf I}}
\nc{\bJ}{{\bf J}} \nc{\bK}{{\bf K}} \nc{\bL}{{\bf L}}
\nc{\bM}{{\bf M}} \nc{\bN}{{\bf N}} \nc{\bO}{{\bf O}}
\nc{\bP}{{\bf P}} \nc{\bQ}{{\bf Q}} \nc{\bR}{{\bf R}}
\nc{\bS}{{\bf S}} \nc{\bT}{{\bf T}} \nc{\bU}{{\bf U}}
\nc{\bV}{{\bf V}} \nc{\bW}{{\bf W}} \nc{\bX}{{\bf X}}
\nc{\bZ}{{\bf Z}} \nc{\bm}{{\bf m}} \nc{\bv}{{\bf v}}
\nc{\ba}{{\bf a}} \nc{\be}{{\bf e}} \nc{\bu}{{\bf u}}
\nc{\brr}{{\bf r}}

\nc{\cA}{{\cal A}} \nc{\cB}{{\cal B}} \nc{\cC}{{\cal C}}
\nc{\cD}{{\cal D}} \nc{\cE}{{\cal E}} \nc{\cF}{{\cal F}}
\nc{\cG}{{\cal G}} \nc{\cH}{{\cal H}} \nc{\cI}{{\cal I}}
\nc{\cJ}{{\cal J}} \nc{\cK}{{\cal K}} \nc{\cL}{{\cal L}}
\nc{\cM}{{\cal M}} \nc{\cN}{{\cal N}} \nc{\cO}{{\cal O}}
\nc{\cP}{{\cal P}} \nc{\cQ}{{\cal Q}} \nc{\cR}{{\cal R}}
\nc{\cS}{{\cal S}} \nc{\cT}{{\cal T}} \nc{\cU}{{\cal U}}
\nc{\cV}{{\cal V}} \nc{\cW}{{\cal W}} \nc{\cX}{{\cal X}}
\nc{\cZ}{{\cal Z}}

\nc{\hA}{{\hat{A}}} \nc{\hB}{{\hat{B}}} \nc{\hC}{{\hat{C}}}
\nc{\hD}{{\hat{D}}} \nc{\hE}{{\hat{E}}} \nc{\hF}{{\hat{F}}}
\nc{\hG}{{\hat{G}}} \nc{\hH}{{\hat{H}}} \nc{\hI}{{\hat{I}}}
\nc{\hJ}{{\hat{J}}} \nc{\hK}{{\hat{K}}} \nc{\hL}{{\hat{L}}}
\nc{\hM}{{\hat{M}}} \nc{\hN}{{\hat{N}}} \nc{\hO}{{\hat{O}}}
\nc{\hP}{{\hat{P}}} \nc{\hR}{{\hat{R}}} \nc{\hS}{{\hat{S}}}
\nc{\hT}{{\hat{T}}} \nc{\hU}{{\hat{U}}} \nc{\hV}{{\hat{V}}}
\nc{\hW}{{\hat{W}}} \nc{\hX}{{\hat{X}}} \nc{\hZ}{{\hat{Z}}}

\nc{\hn}{{\hat{n}}}

\def\max{\mathop{\rm max}}
\def\min{\mathop{\rm min}}

\def\supp{\mathop{\rm supp}}
\def\tr{\mathop{\rm Tr}}

\def\wt{\mathop{\rm wt}}
\def\supp{\mathop{\rm supp}}

\newcommand{\bra}[1]{\langle#1|}
\newcommand{\ket}[1]{|#1\rangle}

\newcommand{\ketbra}[2]{|#1\rangle\!\langle#2|}
\newcommand{\braket}[2]{\langle#1|#2\rangle}

\newcommand{\fl}[2]{\left\lfloor\frac{#1}{#2}\right\rfloor}

\begin{document}

\title{Bounds on $k$-Uniform Quantum States}

	\author{Fei Shi, Yu Ning, Qi Zhao and Xiande Zhang
	
	\thanks{Fei Shi is with Department of Computer Science, University of Hong Kong, Hong Kong, 999077, China    }
	
	\thanks{Yu Ning is with Hefei National Laboratory, University of Science and Technology of China, Hefei 230088, China}
	
	\thanks{Qi Zhao is with Department of Computer Science, University of Hong Kong,   Hong Kong, 999077, China.}
	
	\thanks{Xiande Zhang is with  School of Mathematical Sciences,
		University of Science and Technology of China, Hefei, 230026, China; and with Hefei National Laboratory, University of Science and Technology of China, Hefei 230088, China (email: drzhangx@ustc.edu.cn).
	}
}

\maketitle

\begin{abstract}
	  Do $N$-partite $k$-uniform states always exist when $k\leq \fl{N}{2}-1$? In this work,  we provide new upper bounds on the parameter $k$ for the existence of $k$-uniform states in $(\bbC^{d})^{\otimes N}$ when $d=3,4,5$, which extend  Rains' bound in 1999 and improve Scott's bound in 2004. Since a $k$-uniform state in  $(\bbC^{d})^{\otimes N}$  corresponds to a pure $((N,1,k+1))_{d}$ quantum error-correcting codes,  we also give new upper bounds on the minimum distance $k+1$ of pure $((N,1,k+1))_d$ quantum error-correcting codes.
   Furthermore, we generalize Scott's bound to heterogeneous systems, and show some non-existence results of absolutely maximally entangled states in $\bbC^{d_1}\otimes(\bbC^{d_2})^{\otimes 2n}$.
\end{abstract}


\section{Introduction}

Mutipartite entanglement plays a key role in quantum information and quantum computing, with numerous applications including quantum error-correcting codes \cite{laflamme1996perfect,knill1997theory}, quantum key distribution \cite{ekert1991quantum,gisin2002quantum,bennett1992quantum}, and quantum teleportation \cite{bennett1993teleporting,bouwmeester1997experimental}. However, it is difficult to classify and quantify  multipartite entanglement. The $k$-uniform state, as a special kind of multipartite entangled state, has attracted widespread attention.

A $k$-uniform state in $N$-partite Hilbert space $(\bbC^{d})^{\otimes N}$ is defined as a normalized column vector, with the property that all its reductions to $k$ parties are maximally mixed. For example, the Greenberger-Horne-Zeilinger (GHZ) state is a $1$-uniform state.
$k$-Uniform states are connected to quantum error-correcting codes (QECCs) \cite{scott2004multipartite}, classical error-correcting codes \cite{feng2017multipartite,k_uniform_masking}, orthogonal arrays \cite{goyeneche2014genuinely,li2019k,zang20193,pang2019two},   Latin squares \cite{goyeneche2015absolutely,goyeneche2018entanglement},  symmetric matrices \cite{feng2017multipartite},  and graph states \cite{helwig2013absolutelygraph}. They also can be used for quantum information masking \cite{k_uniform_masking}, quantum secret sharing  \cite{helwig2012absolute}, and holographic quantum
codes \cite{pastawski2015holographic}.
There exists a trivial upper bound on the parameter $k$ for the existence of $k$-uniform states in $(\bbC^{d})^{\otimes N}$, i.e. $k\leq \fl{N}{2}$ \cite{scott2004multipartite}.
Specially, a $\fl{N}{2}$-uniform state is also called  an absolutely maximally
entangled (AME) state \cite{helwig2012absolute}.   However, AME states are very rare for given local dimensions. For example, when $d=2$,  AME states exist only for  $2$-, $3$-, $5$-, and $6$-qubits
\cite{scott2004multipartite,rains1999quantum,huber2017absolutely}. The existence  of AME states in $(\bbC^6)^{\otimes 4}$ was listed as  one of the five open problems in quantum information \cite{horodecki2020five}, which was solved recently \cite{rather2021thirty}. 

A $k$-uniform state in $(\bbC^{d})^{\otimes N}$ corresponds to a pure $((N,1,k+1))_d$ QECC \cite{scott2004multipartite}.  According to this correspondence,  Rains \cite{rains1999quantum} gave an upper bound on the parameter $k$ for the existence of $k$-uniform states in $(\bbC^{2})^{\otimes N}$ (see Theorem~\ref{rain}), which improved the trivial upper bound. For general local dimension $d$, Scott \cite{scott2004multipartite} improved the trivial upper bound by showing the non-existence of AME states (see Theorem~\ref{scott}). When $d=3,4,5$, Huber \emph{et al.} gave the non-existence of more AME states \cite{huber2018bounds}. Note that when $d=2$, Scott's bound is weaker than Rains' bound, since Rains' bound not only shows the non-existence of more AME states, but also shows the non-existence of $k$-uniform states for some $k\leq \fl{N}{2}-1$. When $d\geq 3$, a tighter upper bound than Scott's bound remains unknown.

It is known that $k$-uniform states can be generalized to heterogeneous systems $\bbC^{d_1}\otimes \bbC^{d_2}\otimes \cdots\otimes \bbC^{d_N}$ whose local dimensions are not all equal \cite{goyeneche2016multipartite}. Recently, there are some constructions on $k$-uniform states in heterogeneous systems \cite{shen2021absolutely,pang2021quantum,shi2022k,feng2023constructions}. However, there are few results for the non-existence of $k$-uniform states in heterogeneous systems. This is because $k$-uniform states in heterogeneous systems are more complicated than those in homogeneous systems.
 In  \cite{shi2022k},  the authors showed the non-existence of some AME states in 9-partite, 11-partite, and 13-partite heterogeneous systems.  We will give more general nonexistence results.

In this paper,  by using a series of
enumerators and  the Fast Zeilberger Package in Mathematica \cite{Fastzeil}, we give some upper bounds on the parameter $k$ for the existence of $k$-uniform states in $(\bbC^{d})^{\otimes N}$ when $d=3,4,5$ (see Theorem~\ref{D3bound} and Tables~\ref{D3}, ~\ref{D4}, \ref{D5}), which improve Scott's bound. Moreover, we generalize Scott's bound to heterogeneous systems (see Theorem~\ref{thm:generalize_scott}), and show some non-existence results of AME states in $\bbC^{d_1}\otimes(\bbC^{d_2})^{\otimes 2n}$ (see Corollary~\ref{cor:AME} and Table~\ref{table:AME}).



	 Our paper is organized as follows. In Section~\ref{pre}, we introduce the concept of $k$-uniform states and enumerators. In Section~\ref{se:bound}, we give some upper bounds on the parameter $k$ for the existence of $k$-uniform states in $(\bbC^{d})^{\otimes N}$ when $d=3,4,5$.  Some non-existence results of AME states in heterogeneous systems are given in Section~\ref{sec:AME}.
  Finally, a conclusion is given in Section \ref{con}.

\section{Preliminaries}\label{pre}
In this section, we recall the concepts of $k$-uniform states and enumerators, and some related results.
\subsection{$k$-uniform states}
For any positive integer $d\geq 2$, we denote $\bbZ_d:=\{0,1,\ldots,d-1\}$ and  $\bbZ_d^N:=\bbZ_d\times \bbZ_d\times \cdots \times \bbZ_d$.  An $N$-partite pure quantum state $\ket{\psi}$ with local dimension $d$ refers to a normalized vector in the Hilbert space $(\bbC^d)^{\otimes N}:=\bbC^d\otimes\bbC^d\otimes\cdots\otimes\bbC^d$, which  can be written as $\ket{\psi}=\sum_{\bu\in \bbZ_{d}^{N}}a_{\bu}\ket{\bu}$, where $a_{\bu}\in \bbC$. The density matrix of $\ket{\psi}$ is defined by $\rho:=\ketbra{\psi}{\psi}=\sum_{\bu,\bu'\in \bbZ_d^N}a_{\bu}\overline{a_{\bu'}}\ket{\bu}\bra{\bu'}$.
 For any subset $A$ of $[N]:=\{1,2,\dots,N\}$, let $A^c$ be the complement of $A$.  For a vector $\bu$ of $\bbZ_d^N$, let $\bu_A$ be the  projection of $\bu$ on $A$. The reduced density matrix of $\ket{\psi}$ to $A$ is defined by  $\rho_A:=\sum_{\bu,\bu'\in \bbZ_d^N}a_{\bu}\overline{a_{\bu'}}\braket{\bu_{A^c}}{\bu'_{A^c}}\ketbra{\bu_A}{\bu'_A}$,  where  $\braket{\bu_{A^c}}{\bu'_{A^c}}=1$ if $\bu_{A^c}=\bu'_{A^c}$ and $0$ otherwise.


\begin{definition}
	A pure state $\ket{\psi}=\sum_{\bu\in \bbZ_{d}^{N}}a_{\bu}\ket{\bu}\in (\bbC^d)^{\otimes N}$ is called a \emph{$k$-uniform state} if the reduction to any $k$-parties is maximally mixed, i.e. for any $k$-subset $A$ of $[N]$, $\rho_{A}=\frac{1}{d^k}\sum_{\bu_A\in\bbZ_d^k}\ketbra{\bu_A}{\bu_A}$. Specially, a $\lfloor\frac{N}{2}\rfloor$-uniform state in $(\bbC^{d})^{\otimes N}$ is called an \emph{absolutely maximally entangled (AME) state}.
\end{definition}



For example, the Bell state
\begin{equation*}
\ket{\psi}=\frac{1}{\sqrt{2}}(\ket{00}+\ket{11}) \in \bbC^{2}\otimes \bbC^{2}
\end{equation*}
 is a $1$-uniform state (AME state).  This is because $\rho_{\{1\}}=\rho_{\{2\}}=\frac{1}{2}\sum_{i\in\bbZ_2}\ketbra{i}{i}$.
A $k$-uniform state in $(\bbC^d)^{\otimes N}$ corresponds to a pure $(N, 1, k+1)_d$ QECC \cite{scott2004multipartite}. According to the Schmidt decomposition of bipartite pure states \cite{Nielsen2011Quantum} or quantum Singletion bound for QECCs \cite{rains1999nonbinary}, there exists a trivial upper bound on the parameter $k$ for the existence of $k$-uniform states in $(\bbC^d)^{\otimes N}$, i.e.
\begin{equation*}
   k\leq \fl{N}{2}.
\end{equation*}
When $d=2$, Rains \cite{rains1999quantum} improved this trivial upper bound.
\begin{theorem}[Rains' bound \cite{rains1999quantum}]\label{rain}
	For a $k$-uniform state in $(\bbC^{2})^{\otimes(6m+\ell)}$, where $m\geq 0$ and $0\leq\ell \leq 5$, it satisfies
\begin{equation*}
 k\leq\left\{
	\begin{array}{lll}
	2m+1,&   &\text{if $\ell<5$};  \\
	2m+2,&   &\text{if $\ell=5$}.
	\end{array}
	\right.
\end{equation*}
\end{theorem}

For general $d$, Scott \cite{scott2004multipartite} improved the trivial upper bound as below.
\begin{theorem} [Scott's bound \cite{scott2004multipartite}]\label{scott}
	For a $k$-uniform state in $(\bbC^d)^{\otimes N}$, if
\begin{equation}\label{scottbound}
 N> \left\{
\begin{array}{lll}
2(d^{2}-1),&   &\text{if $N$ is even};  \\
2d(d+1)-1,&   &\text{if $N$ is odd}.
\end{array}
\right.
\end{equation}
then $k\leq\fl{N}{2}-1$.
\end{theorem}

Scott's bound means that AME states do not exist in $(\bbC^d)^{\otimes N}$ when Eq.~\eqref{scottbound} holds.
For $d=3,4,5$, Huber \emph{et al.} gave the non-existence of more AME states \cite{huber2018bounds}, i.e.  $k\leq\fl{N}{2}-1$, when $d=3$ and $N = 8, 12, 13, 14, 16, 17, 19, 21, 23$; $d=4$ and $N = 12, 16, 20, 24, 25, 26, 28, 29, 30, 33, 37, 39$; $d=5$ and $N = 28, 32, 36, 40, 44, 48$.
Note that when $d=2$, Scott's bound is weaker than Rains' bound. According to Scott's bound,  a $k$-uniform state in $(\bbC^2)^{\otimes N}$ must satisfy   $k\leq \fl{N}{2}-1$ for  $N\geq  8$ and $N\neq 9, 11$, while  $k\leq  \fl{N}{2}-1$ for $N=4$ and $N \geq 8$ by Rains' bound.
Moreover, Rains' bound not only shows the non-existence of some AME states, but also shows the non-existence of some $k$-uniform states for $k\leq \fl{N}{2}-1$. For example, Rains' bound shows that $4$-uniform states in $(\bbC^{2})^{\otimes 10}$  do not exist.

\subsection{Enumerators}
 Let $\{e_{j}\}_{j\in \bbZ_{d^2}}$ $(e_{0}=I)$ be an  orthonormal basis acting on $\bbC^{d}$, such that $\tr(e_{i}^{\dag}e_{j})=\delta_{ij}d$. If $j\neq 0$, $\tr(e_{j})=\tr(e_{0}^{\dag}e_{j})=\delta_{0j}d=0$. This basis is also called a \emph{nice error basis} \cite{klappenecker2002beyond}. When $d=2$, $\{e_{j}\}_{j=0}^{3}$ can be built from the Pauli operators; when $d\geq 3$, $\{e_{j}\}_{j\in \bbZ_{d^2}}$ can be built from the  generalized Pauli operators \cite{bertlmann2008bloch}. For the Hilbert space $(\bbC^{d})^{\otimes N}$, a local error basis $\mathcal{E}=\{E_{\alpha}\}$ consists of
 \begin{equation*}
 E_{\alpha}=e_{\alpha_{1}}\otimes e_{\alpha_{2}}\otimes \cdots \otimes e_{\alpha_{N}},
 \end{equation*}
 where $\alpha=(\alpha_{1},\alpha_{2},\cdots,\alpha_{N})\in\bbZ_{d^2}^N$, each $e_{\alpha_{i}}$ acts on $\bbC^{d}$, and $\tr(E_{\alpha}^{\dag}E_{\beta})=\delta_{\alpha\beta}d^{n}$. The support of a local error operator $E_{\alpha}$ is defined as
 $\supp(E_\alpha):=\{i\mid \alpha_i\neq 0, \ 1\leq i\leq N\}$, and the weight of $E_{\alpha}$ is define as  $\wt(E_{\alpha}):=|\supp(E_\alpha)|$. For a subset $A\subset [N]$, we denote $E_{\alpha}^A=\otimes_{i\in A}e_{\alpha_i}$. Specially, let $E_{\alpha}^{\supp}:=E_{\alpha}^{\supp(E_\alpha)}=\otimes_{i\in \supp(E_\alpha)}e_{\alpha_i}$.

For a pure state $\ket{\psi}$ in  $(\bbC^{d})^{\otimes N}$, let $\rho=\ketbra{\psi}{\psi}$ be the density matrix of $\ket{\psi}$. The Shor-Laflamme enumerator \cite{shor1997quantum}  is
\begin{equation*}
    a_j=\sum_{\wt(E_\alpha)=j, E_\alpha\in \cE}\tr(E_\alpha\rho)\tr(E_\alpha^{\dagger}\rho).
\end{equation*}
The corresponding enumerator polynomial is
\begin{equation}\label{Axy}
    A(x,y)=\sum_{j=0}^{N} a_jx^{N-j}y^j.
\end{equation}
Moreover, the quantum MacWilliams identity \cite{shor1997quantum,rains1998quantum,huber2018bounds} is
\begin{equation}\label{eq:AA}
    A(x,y)=A\left(\frac{x+(d^{2}-1)y}{d},\frac{x-y}{d}\right).
\end{equation}
Note that
\begin{equation*}
a_j=\sum_{\wt(E_\alpha)=j, E_\alpha\in \cE}\bra{\psi}E_\alpha\ket{\psi}\bra{\psi}E_\alpha^{\dagger}\ket{\psi}=\sum_{\wt(E_\alpha)=j, E_\alpha\in \cE}|\bra{\psi}E_\alpha\ket{\psi}|^2\geq 0,
\end{equation*}
where $0\leq j\leq N$. Specially, $a_0=|\braket{\psi}{\psi}|^2=1$.

Next, we introduce the shadow enumerator.
The shadow inequality \cite{rains2000polynomial} is as follows. For any fixed subset $T\subset [N]$,
\begin{equation*}
    \sum_{S\subset [N]}(-1)^{|S\cap T|}\tr(\rho_S^2)\geq 0.
\end{equation*}
The shadow enumerator \cite{rains1999quantum,huber2018bounds} is
\begin{equation}\label{eq:s_j1}
   s_j=\sum_{|T|=j}\sum_{S\subset [N]}(-1)^{|S\cap T^c|}\tr(\rho_S^2),
\end{equation}
and the corresponding enumerator polynomial is
\begin{equation}\label{eq:s_j}
    S(x,y)=\sum_{j=0}^Ns_jx^{N-j}y^j,
\end{equation}
where $s_j\geq 0$ for $0\leq j\leq N$.
The shadow enumerator polynomial  and the Shor–Laflamme enumerator polynomial have the following relation \cite{rains1999quantum,huber2018bounds},
\begin{equation*}
    S(x,y)=A\left(\frac{(d-1)x+(d+1)y}{d},\frac{y-x}{d}\right).
\end{equation*}
Moreover,
\begin{equation*}
\begin{aligned}
S(x,y)&=A\left(\frac{(d-1)x+(d+1)y}{d},\frac{y-x}{d}\right)\\
&=A\left(\frac{\frac{(d-1)x+(d+1)y}{d}+(d^{2}-1)\frac{y-x}{d}}{d},\frac{\frac{(d-1)x+(d+1)y}{d}-\frac{y-x}{d}}{d}\right)\\
&=A\left(\frac{-(d-1)x+(d+1)y}{d},\frac{y+x}{d}\right)\\
&=S(-x,y),
\end{aligned}
\end{equation*}
where the second equality is from Eq.~\eqref{eq:AA}. Then $S(x,y)=S(-x,y)$ implies that $s_{N-j}=0$ when $j$ is odd.

Above all, for a pure state $\ket{\psi}$ in $(\bbC^d)^{\otimes N}$, the following
constraints must hold \cite{shor1997quantum,rains1998quantum,rains1999quantum,rains2000polynomial,huber2018bounds}: 
\begin{align}
&A(x,y)=A\left(\frac{x+(d^{2}-1)y}{d},\frac{x-y}{d}\right);\label{eq4}\\
&a_0=1; \label{eq2}\\
&a_{j}\geq 0, \quad \text{for all $1\leq j\leq N$};\label{eq3}\\
&S(x,y)=A\left(\frac{(d-1)x+(d+1)y}{d},\frac{y-x}{d}\right);\label{eq5}\\
&s_j\geq 0, \quad \text{for all $0\leq j\leq N$};\label{eq:s_jgeq}\\
&s_{N-j}=0, \quad \text{when $j$ is odd, for all $0\leq j \leq N$.}\label{eq:s_jeq}
\end{align}
Specially, if $\ket{\psi}$ is a $k$-uniform state in $(\bbC^d)^{\otimes N}$,
then there is another constraint \cite{shor1997quantum,rains1998quantum,huber2018bounds}:
\begin{equation}\label{eq1}
    a_j=0, \quad \text{for all $1\leq j\leq k$.}
\end{equation}
This is because $a_j=\sum_{\wt(E_\alpha)=j, E_\alpha\in \cE}\tr(E_\alpha\rho)\tr(E_\alpha^{\dagger}\rho)=\sum_{\wt(E_\alpha)=j, E_\alpha\in \cE}|\tr(E_\alpha^{\supp}\rho_{\supp(E_\alpha)})|^2
=\sum_{\wt(E_\alpha)=j, E_\alpha\in \cE}$ \par
\noindent $|\frac{1}{d^j}\tr(E_\alpha^{\supp})|^2=0$ for $1\leq j\leq k$.

\section{New upper bounds on the parameter $k$ for the existence of $k$-uniform states in $(\bbC^{d})^{\otimes N}$}\label{se:bound}
In this section, we extend the method used in \cite{rains1999quantum,rains1998shadow} for deducing Rains' bound in Theorem~\ref{rain} from $d=2$ to $d=3,4,5$.
Firstly, we give another form of $A(x,y)$ and $S(x,y)$.

\begin{lemma}\label{thm:axy}
For a pure state $\ket{\psi}$ in $(\bbC^d)^{\otimes N}$, there exist coefficients $c_{i}\in \bbR$, $0\leq i \leq \fl{N}{2}$ such that
\begin{equation*}
\begin{aligned}
A(x,y)&=\sum_{i=0}^{\fl{N}{2}}c_{i}\left(x+(d-1)y\right)^{N-2i}\left(y(x-y)\right)^{i},\\
   S(x,y)&=\sum_{i=0}^{\fl{N}{2}}(-1)^{i}c_{i}2^{N-2i}d^{-i}y^{N-2i}\left(x^{2}-y^{2}\right)^{i}.
\end{aligned}
\end{equation*}
\end{lemma}
\begin{proof} By Eq.~(\ref{eq4}), $A(x,y)=A\left(\frac{x+(d^{2}-1)y}{d},\frac{x-y}{d}\right)$, which means that $A(x,y)$ is invariant under the transformation:
\begin{equation*}
T:
\begin{pmatrix}
        x\\
        y
\end{pmatrix}
  \rightarrow  A\begin{pmatrix}
        x\\
        y
\end{pmatrix},
\end{equation*}
where $A=\frac{1}{d}\begin{pmatrix}
        1 &d^2-1\\
        1  &-1
\end{pmatrix}$. Since $A^2=\bbI$, $A(x,y)$ must be invariant under the group $\{A,\bbI\}$.  By  \cite[Theorem 5, pp 605]{macwilliams1977theory}, 
$A(x,y)$ is a polynomial in $x+(d-1)y$ and $y(x-y)$, i.e.,
\begin{equation*}
   A(x,y)=\sum_{i,j\geq 0}c_{ij}(x+(d-1)y)^{j}(y(x-y))^{i}.
\end{equation*}
Recall from Eq.~\eqref{Axy}, $A(x,y)=\sum_{k=0}^{N}a_{k}x^{N-k}y^{k}$, $a_{k}\geq 0$. Since the degree of every monomial in $A(x,y)$ is $N$,  we have $j+2i=N$. This implies $0\leq i\leq \fl{N}{2}$. Thus
\begin{equation}\label{eqq}
A(x,y)=\sum_{i=0}^{\fl{N}{2}}c_{i}(x+(d-1)y)^{N-2i}(y(x-y))^{i},
\end{equation}
where $c_{i}\in \bbR$ due to $a_{k}\in \bbR$.

By Eq.~(\ref{eq5}) and Eq.~(\ref{eqq}), we have
\begin{equation*}
S(x,y)=A\left(\frac{(d-1)x+(d+1)y}{d},\frac{y-x}{d}\right)
=\sum_{i=0}^{\fl{N}{2}}(-1)^{i}c_{i}2^{N-2i}d^{-i}y^{N-2i}\left(x^{2}-y^{2}\right)^{i}.
\end{equation*}
\end{proof}
\vspace{0.4cm}

According to Eqs.~\eqref{eq:s_j}, \eqref{eq:s_jgeq}, and \eqref{eq:s_jeq}, we know that $S(x,y)=\sum_{j=0}^{N}s_{j}x^{N-j}y^{j}$, where  $s_j\geq 0$, and $s_{N-j}=0$ when $j$ is odd for $0\leq j\leq N$, then we can rewrite:
\begin{equation*}
S(x,y)=\sum_{j=0}^{\fl{N}{2}}b_{j}x^{N-(2j+t)}y^{2j+t},\quad t\equiv N \pmod 2,
\end{equation*}
where $b_{j}=s_{2j+t}\geq 0$ for $0\leq j\leq \fl{N}{2}$. Thus for a pure state in $(\bbC^d)^{\otimes N}$, by setting $x=1$, we have:
\begin{align}
A(1,y)&=\sum_{j=0}^{N}a_{j}y^{j}=\sum_{i=0}^{\fl{N}{2}}c_{i}(1+(d-1)y)^{N-2i}(y(1-y))^{i},\label{A1}\\
S(1,y)&=\sum_{j=0}^{\fl{N}{2}}b_{j}y^{2j+t}=\sum_{i=0}^{\fl{N}{2}}(-1)^{i}c_{i}2^{N-2i}d^{-i}y^{N-2i}\left(1-y^{2}\right)^{i},\label{A2}
\end{align}
where $t\equiv N \pmod 2$, $a_{0}=1$, $a_{j},b_{j}\geq 0$.

By comparing the coefficient of $y^j$ on both sides of Eq.~\eqref{A1} for $0\leq j\leq \fl{N}{2}$, we have
\begin{equation*}
a_j=\sum_{i=0}^j\left(\sum_{\ell=0}^{j-i}\binom{N-2i}{\ell}(d-1)^\ell\binom{i}{j-i-\ell}(-1)^{j-i-\ell}\right)c_i, \\
\end{equation*}
for  $0\leq j\leq \fl{N}{2}$, where $\binom{0}{0}=1$, and $\binom{i}{j}=0$ if $i<j$.
That is
\begin{equation}\label{eq:a_0c_0}
\left\{
\begin{aligned}
a_0&=c_0,\\
a_1&=\binom{N}{1}(d-1)c_0+c_1,\\
a_2&=\binom{N}{2}(d-1)^2c_0+\left(\binom{N-2}{1}(d-1)-1\right)c_1+c_2,\\
&\ldots\\
a_{\fl{N}{2}}&=\sum_{i=0}^{\fl{N}{2}}\left(\sum_{\ell=0}^{{\fl{N}{2}}-i}\binom{N-2i}{\ell}(d-1)^\ell\binom{i}{{\fl{N}{2}}-i-\ell}(-1)^{{\fl{N}{2}}-i-\ell}\right)c_i,\\
\end{aligned}
\right.
\end{equation}
which implies that $c_{i}$ is a linear combination of the coefficients $a_{j}$ for $0\leq j\leq i$, i.e.,
\begin{equation}\label{eq:alpha_i}
    c_{i}=\sum_{j=0}^{i}\alpha_{ij}a_{j}, \quad 0\leq i\leq \fl{N}{2}, \text{ for some real $\alpha_{ij}$}.
\end{equation}
We denote $\alpha_i(N)=\alpha_{i0}$, then we have the following result.

\begin{lemma}\label{alphain}
$\alpha_0(N)=1$, and $\alpha_i(N)=-\frac{N(d-1)}{i}  \sum_{j=0}^{i-1}(1-d)^{j}\binom{N-2i+j}{N-2i}\binom{2i-2-j}{i-1}$, for
	$1\leq i \leq \fl{N}{2}$.
\end{lemma}

The proof of Lemma~\ref{alphain}  is given in Appendix~\ref{appendix:alphain}. Similarly, by comparing the coefficient of  $y^{2j+t}$ on both sides of Eq.~\eqref{A2} for $0\leq j\leq \fl{N}{2}$,
we have
\begin{equation}\label{eq:newbj}
\begin{aligned}
    b_j&=\sum_{i=0}^{\fl{N}{2}}2^{N-2i}d^{-i}\binom{i}{i+j-\fl{N}{2}}(-1)^{j-\fl{N}{2}}c_i\\
    &=\sum_{i=0}^{\fl{N}{2}}2^{N-2i}d^{-i}\binom{i}{\fl{N}{2}-j}(-1)^{\fl{N}{2}-j}c_i\\
     &=\sum_{i=\fl{N}{2}-j}^{\fl{N}{2}}2^{N-2i}d^{-i}\binom{i}{\fl{N}{2}-j}(-1)^{\fl{N}{2}-j}c_i.\\
    &=\sum_{m=0}^{j}2^{2m+t}d^{-(\fl{N}{2}-m)}\binom{\fl{N}{2}-m}{\fl{N}{2}-j}(-1)^{\fl{N}{2}-j}c_{\fl{N}{2}-m},
\end{aligned}
\end{equation}
for $0\leq j\leq \fl{N}{2}$, which implies that $c_i$ is a linear combination of the coefficients $b_{j}$ for $0\leq j\leq \fl{N}{2}-i$, i.e.
\begin{equation*}
    c_{i}=\sum_{j=0}^{\fl{N}{2}-i}\beta_{ij}b_{j}, \quad 0\leq i\leq \fl{N}{2}.
\end{equation*}
We can give the concrete expression of $c_i$ under $b_j$.


\begin{lemma}\label{bij}
$c_{i}=(-1)^{i}2^{2i-N}d^{i}\sum_{j=0}^{\fl{N}{2}-i}\binom{\fl{N}{2}-j}{i}b_{j}$, for $0\leq i\leq\fl{N}{2}$.
\end{lemma}

The proof of Lemma~\ref{bij} is given in Appendix \ref{appendix:bij}. Now, we have the following result.

\begin{lemma}\label{lemma:main_method}
 For a $k$-uniform state in $(\bbC^{d})^{\otimes N}$, we have $k\leq i-1$ if $(-1)^i\alpha_{i}(N)<0$ with $1\leq i\leq \fl{N}{2}$.
\end{lemma}
\begin{proof}
    If $k\geq i$, then $a_0=1$ and $a_1=a_2=\cdots=a_i=0$ by Eqs.~\eqref{eq2} and \eqref{eq1}. According to Eq.~\eqref{eq:alpha_i},
    \begin{equation*}
     c_{i}=\sum_{j=0}^{i}\alpha_{ij}a_{j}=\alpha_{i0}=\alpha_{i}(N).
    \end{equation*}
By Lemma~\ref{bij}, we know that $c_i\leq 0$ when $i$ is odd; $c_i\geq 0$ when $i$ is even.  Therefore, if  $\alpha_{i}(N)>0$ when $i$ is odd, then $c_i=\alpha_{i}(N)>0$, which is impossible. This implies $k\leq i-1$. Similarly,  if  $\alpha_{i}(N)<0$ when $i$ is even, then $k\leq i-1$.
\end{proof}
\vspace{0.4cm}


So we just calculate $\alpha_{i}(N)$ through Lemma~\ref{alphain} to obtain  upper bounds on the parameter $k$ for the existence of  $k$-uniform states in $(\bbC^{d})^{\otimes N}$, where $1\leq i \leq \fl{N}{2}$.
 When $d=2$,  Rains' bound can be obtained by this method, and see \cite{rains1998shadow,rains1999quantum} for more details.
For $d=3$ and small $N$, we obtain Table~\ref{D3}. Note that the  bounds for $N=14,19,23$ are obtained by the non-existence of AME states in $(\bbC^{d})^{\otimes N}$ \cite{huber2018bounds}. For large $N$, we have the following bound. Here $[a,b]$ denotes the set of integers $\{a,a+1,\ldots, b\}$ for any integers $a\leq b$.
\begin{theorem}\label{D3bound}
	 For a $k$-uniform state in $(\bbC^3)^{\otimes N}$,
\begin{equation*}
 k\leq\left\{
\begin{array}{lll}
6m-1,  &   &N\in [14m-4,14m-1];\\
6m+1,&   &N\in [14m,14m+4];\\
6m+3,&   &N\in [14m+5,14m+9],
\end{array}
\right.
\end{equation*}
where $m\geq 1$, except for $N=23,37,51$.
\end{theorem}
\begin{proof} For each $N\in [14m-4,14m-1]$, it suffices to show that $\alpha_{6m}(14m+\ell)<0$ for all $m\geq 1$ and $-4\leq \ell \leq -1$.
The other two cases are the same. When $\ell=-1$, according to Lemma~\ref{alphain}, we have 
\[\alpha_{6m}(14m-1)=-\left(\frac{14m-1}{3m}\right)\sum_{i=0}^{6m-1}(-2)^{i}\binom{i+2m-1}{2m-1}\binom{12m-2-i}{6m-1}.\]
Let $n=2m$, and
\begin{equation*}
p_n=\sum_{i=0}^{3n-1}(-2)^{i}\binom{i+n-1}{n-1}\binom{6n-2-i}{3n-1}.
\end{equation*}
By using the Fast Zeilberger Package in Mathematica \cite{Fastzeil}, we obtain a recursion of $p_n$,
$$h_{n+2}p_{n+2}=h_{n+1}p_{n+1}+h_{n}p_{n},$$
where
\begin{align*}
h_{n+2}=&729 n (1 + n) (1 + 3 n) (2 + 3 n) (4 + 3 n) (5 + 3 n) (148 + 870 n +1247 n^2),\\
h_{n+1}=&48 n (1 + 3 n) (2 + 3 n) (64845 + 611964 n + 2158381 n^2 +3177680 n^3 + 2083331 n^4 + 496306 n^5),\\
h_{n}=&448 (-1 + 7 n) (1 + 7 n) (2 + 7 n) (3 + 7 n) (4 + 7 n) (5 +7 n) (2265 + 3364 n + 1247 n^2).
\end{align*}
Note that $h_n,h_{n+1},h_{n+2}$ are all positive when $n\geq 2$, $p_1=4$ and $p_2=36$. This implies that $p_n>0$ for $n\geq 2$. Thus $\alpha_{6m}(14m-1)<0$ for $m\geq 1$. See Appendix \ref{V111:appen} for the proof of  other cases.  	
\end{proof}
\vspace{0.4cm}

\begin{table}[t]
\caption{Upper bound on the parameter $k$ for the existence of $k$-uniform states in $(\bbC^{3})^{\otimes N}$}
 \centering \label{D3}
\begin{tabular}{|c|c|c|c|c|c|c|c|c|}
	\hline
	$N$                    & 2-3   & 4-5   & 6-8   & 9     & 10-13 & 14    &15-18  & 19    \\ \hline
	$k\leq$              & 1     & 2     & 3     & 4     & 5     & 6     & 7     & 8    \\ \hline
	$N$                    & 20-22 & 23    & 24-27 & 28-32 & 33-36 & 37-41 & 42-46 & 47-50  \\ \hline
	$k\leq$              & 9    & 10    & 11    & 13    & 15    & 17    & 19    & 21       \\ \hline
	$N$                    & 51-55 & 56-60 & 61-65 & 66-69 & 70-74 & 75-79 & 80-83 & 84-88 \\ \hline
	$k\leq$              & 23    & 25    & 27    & 29    & 31    & 33    & 35    & 37    \\ \hline
\end{tabular}
\end{table}


By Theorem~\ref{D3bound}, the upper bound of $k$ is about  $\frac{3}{7}N$ for a  $k$-uniform state in $(\bbC^3)^{\otimes N}$ to exist. This bound is going away from the trivial bound $\fl{N}{2}$  with the increase of $N$. For $d=4,5$ and small $N$,  we list the upper bounds of $k$ in  Tables~\ref{D4} and \ref{D5} by similar arguments. It seems that the bound behaves similarly when $d=4$. However, the computation is so complicated that the Fast Zeilberger Package in Mathematica is not helpful to give a tidy formula as in Theorem~\ref{D3bound}.
Thus, for large $N$, we  have the following conjecture.
\begin{conjecture}\label{conjecture}
\begin{enumerate}[(i)]
\item
	 For a $k$-uniform state in $(\bbC^4)^{\otimes N}$,
\begin{equation*}
 k\leq\left\{
\begin{array}{lll}
8m-5,  &   &N\in [17m-12,17m-9];\\
8m-3,&   &N\in [17m-8,17m-5];\\
8m-1,&   &N\in [17m-4,17m-1];\\
8m+1,&   &N\in [17m,17m+4],
\end{array}
\right.
\end{equation*}
where $m\geq 2$, except for $N=38$.
\item  For a $k$-uniform state in $(\bbC^5)^{\otimes N}$,
\begin{equation*}
\begin{array}{lll}
 k\leq 2m-1,  &   &N\in [4m,4m+3]\\
\end{array}
\end{equation*}
 where $m\geq 45$.
\end{enumerate}

\end{conjecture}

\begin{table}[t]
\caption{Upper bound on the parameter $k$ for the existence of $k$-uniform states in $(\bbC^{4})^{\otimes N}$}
\centering \label{D4}
\begin{tabular}{|c|c|c|c|c|c|c|c|c|}
	\hline
	$N$   & 60-63  &64-67  &68-72  &73-76  &77-80 &81-84  &85-89  &90-93  \\ \hline
	$k\leq$   &29  &31  &33  &35  &37   &39  &41   &43      \\ \hline
	$N$   &94-97&98-102&103-106&107-110 &111-114 &115-119 &120-123&124-127\\ \hline
	$k\leq$ &45    &47    &49   &51   &53    &55   &57   &59    \\ \hline
	$N$   &128-131&132-136&137-140&141-144&145-149&150-153&154-157&158-161 \\ \hline
	$k\leq$ &61 &63  &65   &67  &69   &71    &73     &75    \\ \hline
\end{tabular}
\end{table}

\begin{table}[t]
\caption{Upper bound on the parameter $k$ for the existence of $k$-uniform states in $(\bbC^{5})^{\otimes N}$}
\centering \label{D5}
\begin{tabular}{|c|c|c|c|c|c|c|c|c|}
	\hline
	$N$ &180-183 &184-187 &188-191 &192-195 &196-199 &200-203 &204-207 &208-211\\ \hline
	$k\leq$ &89  &91   &93   &95    &97  &99     &101     &103     \\ \hline
	$N$  &212-215 &216-219 & 220-223 &224-228 & 229-232 & 233-236 &237-240 & 241-244  \\ \hline
	$k\leq$   &105  &107  &109    &111   &113   & 115 &117  &119   \\ \hline
	$N$  & 245-248 & 249-252 & 253-256 & 257-260 &261-264 &265-268 &269-272 &273-276 \\ \hline
	$k\leq$   &121  &123   &125   &127  &129  &131  &133 &135    \\ \hline
\end{tabular}
\end{table}

  Note that our bounds are tighter than Scott's bound in Theorem~\ref{scott} when $d=3,4,5$.  For example, for a  $k$-uniform state in $(\bbC^{3})^{\otimes 88}$, $k\leq 37$ from Table~\ref{D3} while $k\leq 43$ from Scott's bound; for a  $k$-uniform state in $(\bbC^{4})^{\otimes 161}$, $k\leq 75$ from Table~\ref{D4} while $k\leq 79$ from Scott's bound; for  a  $k$-uniform state in $(\bbC^{5})^{\otimes 276}$, $k\leq 135$ from Table~\ref{D5} while $k\leq 137$ from Scott's bound.

\section{non-existence of AME states in heterogeneous systems}\label{sec:AME}
In this section, we generalize Scott's bound to heterogeneous systems whose local dimensions are not all equal, and give some non-existence results of  AME states in heterogeneous systems. We also compare the results obtained by the generalized Scott's bound and shadow inequality.

Note that there are two definitions of AME states in heterogeneous systems $\bbC^{d_1}\otimes\bbC^{d_2}\otimes\cdots\otimes\bbC^{d_N}$. The first one requires the reduction to any $\fl{N}{2}$ parties is maximally mixed \cite{goyeneche2016multipartite}, while the second one requires  every subsystem whose dimension is not larger than
that of its complement must be maximally mixed \cite{huber2018bounds}. We only focus on the first one in this paper.
 According to the Schmidt decomposition of bipartite pure states  \cite{Nielsen2011Quantum},  an AME state in $\bbC^{d_1}\otimes\bbC^{d_2}\otimes\cdots\otimes\bbC^{d_N}$  must satisfy:
for any  $\fl{N}{2}$-subset  $A$  of $[N]$,   $\prod_{i\in A}d_i\leq \prod_{i\in A^c}d_i$.
Specially, if $d_1, d_2,\ldots, d_N$ are not all equal, then  $N$ must be odd.
Next, we  generalize Scott's bound to heterogeneous systems.
\begin{theorem}[Generalized Scott's bound]\label{thm:generalize_scott}
AME states  in $\bbC^{d_1}\otimes \bbC^{d_2}\otimes \cdots \otimes \bbC^{d_N}$ do not exist if there exists a  $(\fl{N}{2}+2)$-subset $A$ of $[N]$, such that
\begin{equation}\label{eq:generalize}
    \frac{\prod_{i\in A}d_i^2}{d_1d_2\cdots d_N}\left(1-\sum_{i\in A}\frac{1}{d_i^2}\right)+\fl{N}{2}+1<0.
\end{equation}
\end{theorem}
\begin{proof}
The proof is motivated by the new proof of Scott's bound in the Supplemental material of \cite{huber2017absolutely}. For completeness, we present every detail here. Firstly,  we  choose an orthonormal basis $\{e_j\}_{j\in\bbZ_{d_i^2}}$ ($e_0=\bbI$) of Hermitian operators (e.g. the generalized Gell-Mann matrix basis \cite{bertlmann2008bloch}) acting on $\bbC^{d_i}$ for $1\leq i\leq N$, which satisfies $\tr(e_{j}e_{j'})=\delta_{jj'}d_i$, and $\tr(e_j)=0$ for $j\neq 0$.
For the Hilbert space $\bbC^{d_1}\otimes\bbC^{d_2}\otimes\cdots\otimes\bbC^{d_N}$, a local error basis  $\cE=\{E_\alpha\}$ consists of $E_{\alpha}=e_{\alpha_1}\otimes e_{\alpha_2}\otimes\cdots \otimes e_{\alpha_N}$, where
$\alpha=(\alpha_1,\alpha_2,\cdots,\alpha_N)\in \bbZ_{d_1^2}\times \bbZ_{d_2^2}\times\cdots\times \bbZ_{d_N^2}$,
each $e_{\alpha_i}$ acts on $\bbC^{d_i}$,  $\tr(E_{\alpha}E_{\beta})=\delta_{\alpha\beta}d_1d_2\cdots d_n$, and $\tr(E_{\alpha})=0$ for $1\leq \wt(E_{\alpha})\leq N$.
For an AME state $\ket{\psi}$ in $\bbC^{d_1}\otimes\bbC^{d_2}\otimes\cdots\otimes\bbC^{d_N}$, let $\rho=\ketbra{\psi}{\psi}$, then  $\rho$ has the Bloch representation \cite{eltschka2015monogamy},
\begin{equation*}
    \rho=\frac{1}{d_1d_2\cdots d_N}\sum_{E_{\alpha}\in \cE}r_{\alpha}E_{\alpha},
\end{equation*}
where $r_{\alpha}=\tr(E_{\alpha}\rho)\in \bbR$. When $1\leq \wt(E_{\alpha})\leq \fl{N}{2}$, $\tr(E_{\alpha}\rho)=\tr(E_{\alpha}^{\supp}\rho_{\supp(E_{\alpha})})=\frac{1}{\prod_{i\in \supp(E_{\alpha})}d_i}\tr(E_{\alpha}^{\supp})=0$  since the reduction to any $\fl{N}{2}$ parties is maximally mixed. Then
\begin{equation*}
    \rho= \frac{1}{d_1d_2\cdots d_N}\left(\bbI+\sum_{\wt(E_{\alpha})\geq \fl{N}{2}+1}r_{\alpha}E_{\alpha}\right).
\end{equation*}
Let $A=\{i_1,i_2,\ldots,i_{k}\}$ be a $k$-subset of $[N]$ with $k\geq \fl{N}{2}+1$, then
\begin{equation}\label{eq:rho_A}
    \rho_{A}=\frac{1}{d_{i_1}d_{i_2}\cdots d_{i_k}}\left(\bbI+\sum_{\fl{N}{2}+1\leq \wt(E_\alpha)\leq k,   \supp(E_{\alpha})\subset A}r_{\alpha}E_{\alpha}^{A}\right).
\end{equation}
Note that Eq.~\eqref{eq:rho_A} is true even if we only require the reduction to any $\fl{N}{2}$-subset of $A$ is maximally mixed.
Since $\tr(E_\alpha^{A}E_\beta^{A})=\delta_{\alpha\beta}d_{i_1}d_{d_2}\cdots d_{i_{k}}$ for any $\supp(E_{\alpha}), \supp(E_{\beta})\subset A$, we obtain
 \begin{equation*}
 \tr(\rho_{A}^2)=\frac{1}{d_{i_1}d_{i_2}\cdots d_{i_k}}\left(1+\sum_{\fl{N}{2}+1\leq \wt(E_\alpha)\leq k,  \supp(E_{\alpha})\subset A}r_{\alpha}^2\right).
 \end{equation*}
Moreover, $\rho_A$ also has the property \cite{huber2017absolutely}:
\begin{equation}\label{eq:rhoA}
   \rho_A^2=\frac{d_{i_1}d_{i_2}\cdots d_{i_k}}{d_1d_2\cdots d_N}\rho_A,
\end{equation}
which is obtained by the property that $\rho_{A^c}$ is maximally mixed.
By Eq.~\eqref{eq:rhoA},
$\tr(\rho_{A}^2)=\frac{d_{i_1}d_{i_2}\cdots d_{i_k}}{d_1d_2\cdots d_N}$. Then we have
\begin{equation*}
 \tr(\rho_{A}^2)=\frac{1}{d_{i_1}d_{i_2}\cdots d_{i_k}}\left(1+\sum_{\fl{N}{2}+1\leq \wt(E_\alpha)\leq k,  \supp(E_{\alpha})\subset A}r_{\alpha}^2\right)=\frac{d_{i_1}d_{i_2}\cdots d_{i_k}}{d_1d_2\cdots d_N}.
\end{equation*}
Hence,
\begin{equation}\label{eq:trrho}
    \sum_{\fl{N}{2}+1\leq \wt(E_\alpha)\leq k,  \supp(E_{\alpha})\subset A}r_{\alpha}^2=\frac{(d_{i_1}d_{i_2}\cdots d_{i_k})^2}{d_1d_2\cdots d_N}-1,
\end{equation}
for any $k$-subset $A=\{i_1,i_2,\ldots,i_{k}\}$ of $[N]$ with $k\geq \fl{N}{2}+1$.

Now let $k=\fl{N}{2}+2$ (i.e. $A=\{i_1,i_2,\ldots,i_{\fl{N}{2}+2}\}$) and $B_j=A\setminus\{i_j\}$ for $1\leq j\leq \fl{N}{2}+2$, then we have
\begin{equation}\label{eq:rgeq}
\begin{aligned}
    \sum_{ \supp(E_{\alpha})=A}r_{\alpha}^2&= \sum_{\fl{N}{2}+1\leq \wt(E_\alpha)\leq \fl{N}{2}+2,  \supp(E_{\alpha})\subset A}r_{\alpha}^2-\sum_{\wt(E_\alpha)= \fl{N}{2}+1,  \supp(E_{\alpha})\subset A}r_{\alpha}^2\\
    &=\sum_{\fl{N}{2}+1\leq \wt(E_\alpha)\leq \fl{N}{2}+2,  \supp(E_{\alpha})\subset A}r_{\alpha}^2-\sum_{j=1}^{\fl{N}{2}+2}\
    \sum_{\supp(E_{\alpha})= B_j}r_{\alpha}^2\geq 0.
    \end{aligned}
\end{equation}
By Eq.~\eqref{eq:trrho} and Eq.~\eqref{eq:rgeq},
we obtain
\begin{equation}
    \left(\frac{\left(d_{i_1}d_{i_2}\cdots d_{i_{\fl{N}{2}+2}}\right)^2}{d_1d_2\cdots d_N}-1\right)
-\sum_{j=1}^{\fl{N}{2}+2}\left(\frac{\left(d_{i_1}d_{i_2}\cdots d_{i_{\fl{N}{2}+2}}\right)^2}{d_1d_2\cdots d_Nd_j^2}-1\right)\geq 0,
\end{equation}
i.e.
\begin{equation}
    \frac{\prod_{i\in A}d_i^2}{d_1d_2\cdots d_N}\left(1-\sum_{i\in A}\frac{1}{d_i^2}\right)+\fl{N}{2}+1\geq 0.
\end{equation}
This completes the proof.
\end{proof}
\vspace{0.4cm}

When $d_1=d_2=\cdots=d_N=d$, Scott's bound in Theorem~\ref{scott} can be obtained by Theorem~\ref{thm:generalize_scott}.
Next, we consider a simple heterogeneous system  $\bbC^{d_1}\otimes (\bbC^{d_2})^{\otimes 2n}$. According to the Schmidt decomposition of bipartite states, $d_1d_2^{n-1}\leq d_2^{n+1}$, i.e. $d_1\leq d_2^2$.
\begin{corollary}\label{cor:AME} Let $d_1\leq d_2^2$. 
  AME states  in $\bbC^{d_1}\otimes (\bbC^{d_2})^{\otimes 2n}$ do not exist if
\begin{enumerate}[(i)]
\item $d_1< d_2$, and $n>\frac{d_2^4-d_1}{d_2^2-d_1}-2$; or 
\item $d_2\leq d_1\leq d_2^2$, and  $n>\frac{d_2^2(d_1+1)}{d_1}-1$.
\end{enumerate}
\end{corollary}
\begin{proof}
 Let $A=\{1,2,\cdots, n+2\}$, then by Theorem~\ref{thm:generalize_scott}, we have
\begin{equation*}
    \frac{d_1^2d_2^{2(n+1)}}{d_1d_2^{2n}}\left(1-\frac{1}{d_1^2}-(n+1)\frac{1}{d_2^2}\right)+n+1<0,
\end{equation*}
i.e.
\begin{equation*}
    n>\frac{d_2^2(d_1+1)}{d_1}-1.
\end{equation*}
Let $A=\{2,3,\ldots,n+3\}$, then by Theorem~\ref{thm:generalize_scott}, we have
\begin{equation}\label{eq:hold1}
    \frac{d_2^{2(n+2)}}{d_1d_2^{2n}}\left(1-(n+2)\frac{1}{d_2^2}\right)+n+1<0,
\end{equation}
i.e.
\begin{equation}
   n>\frac{d_2^4-d_1}{d_2^2-d_1}-2,
\end{equation}
where $d_1<d_2^2$.
 By Theorem~\ref{thm:generalize_scott},  AME states  in $\bbC^{d_1}\otimes (\bbC^{d_2})^{\otimes 2n}$ do not exist  if $d_1<d_2^2$, and $n>\min\{\frac{d_2^2(d_1+1)}{d_1}-1, \frac{d_2^4-d_1}{d_2^2-d_1}-2\}$.
Note that when $d_1<d_2$, we have
$\frac{d_2^4-d_1}{d_2^2-d_1}-2< \frac{d_2^2(d_1+1)}{d_1}-1$.
When $d_2\leq d_1< d_2^2$, we have $\frac{d_2^4-d_1}{d_2^2-d_1}-2\geq  \frac{d_2^2(d_1+1)}{d_1}-1$. When $d_1=d_2^2$, Eq.~\eqref{eq:hold1} does not hold, which means that  AME states  in $\bbC^{d_1}\otimes (\bbC^{d_2})^{\otimes 2n}$ do not exist if $d_1=d_2^2$, and  $n> \frac{d_2^2(d_1+1)}{d_1}-1$.
This completes the proof.
\end{proof}
\vspace{0.4cm}

\renewcommand\arraystretch{1.5}
\begin{table}[t]
	\caption{Some non-existence results for AME states in $\bbC^{d_1}\otimes (\bbC^{d_2})^{\otimes 2n}$, where $d_1\leq d_2^2$}
	\centering \label{table:AME}
	\begin{tabular}{|c|c|c|}
		\hline	
		$\bbC^{d_1}\otimes (\bbC^{d_2})^{\otimes 2n}$ & Corollary~\ref{cor:AME}  & Shadow inequality of Eq.~\eqref{eq:shadow}   \\
		\hline
		 $3\leq d_1\leq 4$, and $d_2=2$ & $n\geq 5$ &$n=4$ \\
		 $d_1=2$, and $d_2=3$ & $n\geq 10$ & $n=6, 8, 9$ \\
    $ d_1= 4$, and $d_2=3$ & $n\geq 11$ & $n=6, 8, 9, 10$\\
		 $5\leq d_1\leq 8$, and $d_2=3$ & $n\geq 10$ & $n=6, 8, 9$\\
        $d_1= 9$, and $d_2=3$ & $n\geq 10$ &  $n=6, 8$\\
    $d_1=2$, and $d_2=4$ & $n\geq 17$ &  $n=10, 12, 14, 16$\\
   $d_1=3$, and $d_2=4$ & $n\geq 18$ &  $n=12, 14, 16$\\
      $d_1=5$, and $d_2=4$ & $n\geq 19$ &  $n=12, 14, 16, 18$\\
         $6\leq d_1\leq 8$, and $d_2=4$ & $n\geq 18$ &  $n=14, 16$\\
    $9\leq d_1\leq 16$, and $d_2=4$ & $n\geq 17$ &  $n=14, 16$\\
		\hline
	\end{tabular}
\end{table}

In \cite{shi2022k}, we have shown some non-existence results of AME states in 9-parties, 11-parties, and 13-parties heterogeneous systems. Now, we can obtain families of non-existence results by Corollary~\ref{cor:AME}.
 See the second column in Table~\ref{table:AME} for some non-existence results of AME states in $\bbC^{d_1}\otimes (\bbC^{d_2})^{\otimes 2n}$ by Corollary~\ref{cor:AME}.

The  shadow inequality of Eq.~\eqref{eq:s_j1} can be also generalized to heterogeneous systems \cite{huber2018bounds,shi2022k}.
 For an AME states in
$\bbC^{d_1}\otimes\bbC^{d_2}\otimes\cdots\otimes\bbC^{d_N}$, $N$ is odd,
the shadow inequality of Eq.~\eqref{eq:s_j1} can be expressed by \cite{huber2018bounds,shi2022k}:
\begin{equation}\label{eq:shadow}
s_{j}=\sum_{k=0}^N\left(\sum_{\alpha=\max\{0, k-j\}}^{\min\{k, N-j\}}(-1)^{\alpha}\binom{N-k}{N-j-\alpha}\binom{k}{\alpha}\right)A_k'\geq 0,
\end{equation}
for all $0\leq j\leq N$,
$A_0'=1$, $A_k'=\sum_{\{i_1,i_2,\ldots,i_k\}\subset\{1,2,\ldots,N\}}\frac{1}{d_{i_1}d_{i_2}\cdots d_{i_k}}$ for $1\leq k\leq \frac{N-1}{2}$, and $A_k'=A_{N-k}'$ for $0\leq k\leq \frac{N-1}{2}$. Now, by using the Fast Zeilberger Package and the shadow inequality,  we can also obtain families of  non-existence results of AME states.
\begin{proposition}\label{thm:AME}
\begin{enumerate}[(i)]
    \item AME states in $\bbC^{d}\otimes(\bbC^{2})^{\otimes 2n}$ do not exist for $d=3,4$ and $n\geq 4$.
    \item AME states in $\bbC^{d}\otimes(\bbC^{3})^{\otimes 2n}$ do not exist for $2\leq d\leq 8$ and $d\neq 3$, $n\geq 6$ and $n\neq 7$; $d=9$,  $n\geq 6$ and $n\neq 7,9$.
\end{enumerate}
\end{proposition}
\begin{proof}
(i)	If $d=3$, then by Eq.~\eqref{eq:shadow},  we obtain that $s_1=2f(n)$, where
	\begin{equation*}
	f(n)=(2n+1)+\sum_{i=1}^{n}(-1)^{i}(2n+1-2i)\left[\binom{2n}{i}\frac{1}{2^i}+\binom{2n}{i-1}\frac{1}{3\times 2^{i-1}}\right],
	\end{equation*}
	and $s_3=2g(n)$, where
	\begin{equation*}
	\begin{aligned}
	g(n)=&\binom{2n+1}{3}+\sum_{i=1}^{n}\left[(-1)^{i-3}\binom{i}{3}+(-1)^{i-2}(2n+1-i)\binom{i}{2}+(-1)^{i-1}\binom{2n+1-i}{2}i+(-1)^i\binom{2n+1-i}{3}\right]\\
	&\left[\binom{2n}{i}\frac{1}{2^i}+\binom{2n}{i-1}\frac{1}{3\times 2^{i-1}}\right].
	\end{aligned}
	\end{equation*}	
	There are two cases. If $n=2m$, then by using the Fast Zeilberger Package in Mathematica \cite{Fastzeil}, we have
	\begin{equation*}
	f(2m+2)=p_m\cdot f(2m)-q_m,
	\end{equation*}
	where
	\begin{equation*}
	\begin{aligned}
	p_m&=\frac{7+6m}{16 (1 + 6 m)},\\
	q_m&=\frac{(2 + m) (5 + 2 m) (13 + 177 m + 704 m^2 +
		420 m^3)\binom{7 + 4 m}{2 + 2 m} \times2^{-1 - 2 m}}{
		48 (1 + 4 m) (3 + 4 m) (5 + 4 m) (7 + 4 m)  (1 + 6 m)}.
	\end{aligned}
	\end{equation*}
	Since  $p_m>0$, $q_m>0$ when $m\geq 2$, and $f(4)=-\frac{23}{24}<0$, we have $f(2m+2)<0$ for $m\geq 2$.
	Thus
	\begin{equation}\label{eq:S1}
	s_1=2f(n)<0 \quad \text{for} \quad n=2m\geq 4.
	\end{equation}
	
	If $n=2m+1$, then by using the Fast Zeilberger Package in Mathematica \cite{Fastzeil}, we have
	\begin{equation*}
	g(2m+3)=u_m\cdot g(2m+1)-v_m,
	\end{equation*}
	where
	\begin{equation*}
	\begin{aligned}
	u_m&=\frac{(3 + 2 m)^2 (5 + 4 m)}{16 (1 + 2 m)^2 (1 + 4 m)},\\
	v_m&=\frac{(3 + m) (1 + 2 m) (3 + 2 m) (5 + 2 m) (21 + 445 m +
		1024 m^2 + 420 m^3) \binom{9+4m}{3 + 2 m}\times2^{-2 - 2 m}}{
		144 (3 + 4 m) (7 + 4 m) (9 + 4 m)  (1 + 2 m)^2 (1 + 4 m)}.
	\end{aligned}
	\end{equation*}
	Since  $u_m>0$, $v_m>0$ when $m\geq 2$, and $g(5)=-\frac{65}{8}<0$, we have $g(2m+3)<0$ for $m\geq 2$.
	Thus
	\begin{equation}\label{eq:S3}
	s_3=2g(n)<0 \quad \text{for} \quad n=2m+1\geq 5.
	\end{equation}
	According to Eqs.~\eqref{eq:shadow}, \eqref{eq:S1} and \eqref{eq:S3}, we know that
	AME states in $\bbC^3\otimes(\bbC^{2})^{\otimes 2n}$ do not exist for $n\geq 4$.
	
If $d=4$, we can also show that AME states in $\bbC^4\otimes(\bbC^{2})^{\otimes 2n}$ do not exist for $n\geq 4$ for the same discussion as above, see Appendix~\ref{appendix:d4} for details.	

(ii) Please see Appendix~\ref{appendix:d4} for the proof.	
\end{proof}
\vspace{0.4cm}

Proposition~\ref{thm:AME} excludes more parameters for the specified systems. Assisted by computer, we obtain the third column of  Table~\ref{table:AME} by shadow inequality. It seems that Corollary~\ref{cor:AME} can be covered by the shadow inequality if computing complexity is not considered. However, this is not true. 
There exist some non-existence results of AME states that are obtained by Corollary~\ref{cor:AME}, while they cannot be obtained by the shadow inequality. For example, AME states in $\bbC^{2}\otimes (\bbC^{4})^{\otimes 34}$ do not exist by Corollary~\ref{cor:AME}, while the shadow inequality shows that $s_i\geq 0$ for $0\leq i\leq 34$.

In addition to AME states, Theorem~\ref{thm:generalize_scott} can also show some non-existence reuslts of more pure states. From the proof of Theorem~\ref{thm:generalize_scott}, we know that
if  there exists a subset $A=\{i_1,i_2,\cdots,i_{\fl{N}{2}+2}\}\subset[N]$ such that the inequality of Eq.~\eqref{eq:generalize}  holds, then any pure state (not necessarily AME state) in $\bbC^{d_1}\otimes \bbC^{d_2}\otimes \cdots \otimes \bbC^{d_N}$ with the following property does not exist: the reduction to any  $\fl{N}{2}$-subset of $A$ is maximally mixed and the reduction to $A^c\cup\{i_j\}$ is maximally mixed for $1\leq j\leq \fl{N}{2}+2$.
Note that $k$-uniform states can also be generalized to heterogeneous systems \cite{goyeneche2016multipartite,shen2021absolutely,pang2021quantum,shi2022k,feng2023constructions}.
It is difficult to give some upper bounds on the parameter $k$ for the existence of $k$-uniform states in $\bbC^{d_1}\otimes\bbC^{d_2}\otimes\cdots\otimes\bbC^{d_N}$ when $k\leq \fl{N}{2}-1$, this is because the quantum MacWilliams identiy  is more complicated in heterogeneous systems  than in homogeneous systems \cite{nebe2006self}.

\section{Conclusion}\label{con}

In this paper, we gave some  upper bounds on the parameter $k$ for the existence of $k$-uniform states in $(\bbC^d)^{\otimes N}$ when $d=3,4,5$  by using enumerators and the Fast Zeilberger Package. We also generalized Scott's bound to heterogeneous systems, and provided some non-existence results of AME states in heterogeneous systems.
However,  there are still a lot of questions left. For example,  how to prove Conjecture~\ref{conjecture}? How to give some upper bounds on the parameter $k$ for the existence of $k$-uniform states in $(\bbC^d)^{\otimes N}$ when $d\geq 6$? How to generalize Rains' bound to heterogeneous systems?

\section*{Acknowledgments}
The authors thank Z.-W. Sun and C. Wang for providing the idea to use the Fast Zeilberger Package. The authors also thank F. Huber for pointing out the two different definitions of AME states in heterogeneous systems.
 Y.N. and X.Z.  acknowledges funding from the National Key Research and Development Program of China
(Grant No. 2020YFA0713100), the NSFC under
Grants No. 12171452 and No. 12231014,  and the Innovation Program for Quantum Science
and Technology (Grant No. 2021ZD0302902).  Q.Z.  acknowledges funding
from HKU Seed Fund for Basic Research for New Staff via
Project 2201100596 and Guangdong Natural Science FundGeneral Programme via Project 2023A1515012185.
\appendix

\subsection{Two Lemmas}\label{appendix:onelemma}

\begin{lemma}[Lagrange inversion theorem \cite{dobrushkin2016methods}]\label{lagrange}
Let $z=f(w)$ be an analytic function at the point $0$, with $f(0)=0$ and $f'(0)\neq 0$, then
\begin{equation*}
   w=\sum_{n=1}^{+\infty}h_nz^n,
\end{equation*}
where $h_n=\frac{1}{n!}\text{lim}_{w\rightarrow 0}\frac{d^{n-1}}{dw^{n-1}}\left[\left(\frac{w}{f(w)}\right)^n\right]$.
\end{lemma}

\begin{lemma}[B\"{u}rmann-Lagrange\cite{rains1998shadow,Modern}]\label{lagrange1}
	Let $f(x)$ and $g(x)$ be formal power series, with $g(0)=0$ and $g'(0)\neq 0$. If coefficients $\kappa_{i}$ are defined by
	\begin{equation*}
	f(x)=\sum_{i\geq 0}\kappa_{i}g(x)^{i},
	\end{equation*} then
	\begin{equation*}
	\kappa_{i}=\frac{1}{i}\left[\text{coeff. of}\ \ x^{i-1}\ \text{in}\ f'(x)\left(\frac{x}{g(x)}\right)^{i}\right].
	\end{equation*}
\end{lemma}

\subsection{The proof of Lemma~\ref{alphain}}\label{appendix:alphain}
\begin{proof}
By Eq.~\eqref{eq:a_0c_0}, we know that $\alpha_0(N)=\alpha_{00}=1$. Since it is not easy to solve the non-homogeneous linear equations of Eq.~\eqref{eq:a_0c_0}, we compute $\alpha_i(N)=\alpha_{i0}$ by using formal power series. In order to compute $\alpha_{i0}$, we  assume $a_1=a_2=\cdots=a_{\fl{N}{2}}=0$.
In this case, $c_{i}=\sum_{j=0}^{i}\alpha_{ij}a_{j}=\alpha_{i0}=\alpha_{i}(N)$ for $1\leq i\leq \fl{N}{2}$.

According to Eq.~(\ref{A1}), we have
\begin{equation*}
\sum_{j=0}^{N}a_{j}y^{j}=1+\sum_{j=\fl{N}{2}+1}^{N} a_{j}y^{j}=\sum_{i=0}^{\fl{N}{2}}c_{i}(1+(d-1)y)^{N-2i}(y(1-y))^{i}.
	\end{equation*}
 Let $f(y)=1+(d-1)y$ and $g(y)=y(1-y)$, then
	\begin{equation*}
	1+\sum_{j=\fl{N}{2}+1}^{N} a_{j}y^{j}=\sum_{i=0}^{\fl{N}{2}}c_{i}f^{N-2i}g^{i}.
	\end{equation*}
	Dividing by $f^{N}$ on both sides, we have
	\begin{equation}\label{eq:f}
	f^{-N}=\sum_{i=0}^{\fl{N}{2}}c_{i}\phi^{i}-f^{-N}\sum_{j=\fl{N}{2}+1}^{N}a_{j}y^{j},
	\end{equation}
	where $\phi=\frac{g}{f^{2}}$. Note that   $f^{-N}=(1+(d-1)y)^{-N}=\sum_{i=0}^{+\infty}\binom{N+i-1}{N-1}((1-d)y)^i$.
 Then Eq.~\eqref{eq:f} can be expressed by
 \begin{equation}\label{eq:f1}
	f^{-N}=\sum_{i=0}^{\fl{N}{2}}c_{i}\phi^{i}+\sum_{j=\fl{N}{2}+1}^{+\infty}\beta_{j}y^{j},
 \end{equation}
Since $\phi(0)=0$ and $\phi'(0)\neq0$, $y=\sum_{n=1}^{+\infty}h_n\phi^n$ by Lemma~\ref{lagrange}. Then Eq.~\eqref{eq:f1} can be expressed by
\begin{equation*}
   	f^{-N}=\sum_{i=0}^{\fl{N}{2}}c_{i}\phi^{i}+\sum_{j=\fl{N}{2}+1}^{+\infty}\gamma_{j}\phi^{j},
\end{equation*}
According to Lemma~\ref{lagrange1}, we have \[c_{i}=-\frac{N(d-1)}{i}\left[\text{coeff. of}\ \ y^{i-1}\ \text{in}\ (1+(d-1)y)^{-(N+1-2i)}(1-y)^{-i}\right].\]
Since $(1+(d-1)y)^{-(N+1-2i)}=\sum_{\ell_1=0}^{+\infty}\binom{N-2i+\ell_1}{N-2i}((1-d)y)^{\ell_1}$ and $(1-y)^{-i}=\sum_{\ell_2=0}^{+\infty}\binom{i-1+\ell_2}{i-1}y^{\ell_2}$ by the formal power series,  we have \[\alpha_{i}(N)=c_i=-\frac{N(d-1)}{i}\sum_{j=0}^{i-1}(1-d)^{j}\binom{N-2i+j}{N-2i}\binom{2i-2-j}{i-1}.\]
\end{proof}

\subsection{The proof of Lemma~\ref{bij}}\label{appendix:bij}
\begin{proof} According to Eq.~\eqref{eq:newbj}, we have
\begin{equation}
    b_{j}=\sum_{m=0}^{j}2^{2m+t}d^{-(\fl{N}{2}-m)}\binom{\fl{N}{2}-m}{\fl{N}{2}-j}(-1)^{\fl{N}{2}-j}c_{\fl{N}{2}-m},
\end{equation}
 where $0\leq j\leq \fl{N}{2}$. Let $p=\fl{N}{2}$, then
	\begin{equation}\label{eq9}
	b_{j}=\sum_{m=0}^{j}(-1)^{p-j}2^{2m+t}d^{-(p-m)}\binom{p-m}{p-j}c_{p-m},
	\end{equation}
	where $0\leq j\leq p$. Let $x_{i}=c_{p+1-i}$ for $1\leq i \leq p+1$, we have:
	$$\left\{
	\begin{aligned}
	&b_{0}=A_{11}x_{1},\\
	&b_{1}=A_{21}x_{1}+A_{22}x_{2},\\
	&b_{2}=A_{31}x_{1}+A_{32}x_{2}+A_{33}x_{3},\\
	&\cdots\\
	&b_{p}=A_{p+1,1}x_{1}+A_{p+1,2}x_{2}+A_{p+1,3}x_{3}+\cdots+A_{p+1,p+1}x_{p+1},
	\end{aligned}
	\right.$$
	where the coefficient matrix $A=(A_{\ell j})$ with
	$$A_{lj}=\left\{
	\begin{array}{lll}
	(-1)^{p+1-\ell}2^{2(j-1)+t}d^{-(p-j+1)}\binom{p-j+1}{p-\ell+1},&   &1\leq j\leq \ell\leq p+1,  \\
	0,&   &1\leq \ell< j\leq p+1.
	\end{array}
	\right.
	$$
	The inverse matrix $A^{-1}=(A^{-1}_{\ell j})$ is given by
	$$A^{-1}_{lj}=\left\{
	\begin{array}{lll}
	(-1)^{p+1-\ell}2^{-2(\ell-1)-t}d^{p-\ell+1}\binom{p-j+1}{p-\ell+1},&   &1\leq j\leq \ell\leq p+1,  \\
	0,&   &1\leq \ell< j\leq p+1.
	\end{array}
	\right.
	$$\\
	This because when $\ell< j$, $(AA^{-1})_{\ell j}=0$; when
	 $\ell \geq j$,
	\begin{align*}
	(AA^{-1})_{\ell j}&=\sum_{i=1}^{p+1}A_{li}A_{ij}^{-1}\\
	&=\sum_{i=j}^{\ell}(-1)^{-i-\ell}\binom{p-j+1}{p-i+1}\binom{p-i+1}{p-\ell+1}\\
	&=\binom{p-j+1}{p-\ell+1}\sum_{i=j}^{\ell}(-1)^{\ell-i}\binom{\ell-j}{\ell-i}\\
	&=\binom{p-j+1}{p-\ell+1}\sum_{t=0}^{\ell-j}(-1)^{t}\binom{\ell-j}{t}\\
	&=\binom{p-j+1}{p-\ell+1}(1-1)^{\ell-j}\\
	&=\delta_{\ell j}.
	\end{align*}
	Thus
	\begin{align*}
	c_{i}=x_{p+1-i}&=\sum_{j=1}^{p+1-i}A_{p+1-i,j}^{-1}b_{j-1}\\
	&=\sum_{j=1}^{p+1-i}(-1)^{i}2^{-2(p-i)-t}d^{i}\binom{p-j+1}{i}b_{j-1}\\
	&=\sum_{j=0}^{p-i}(-1)^{i}2^{-2(p-i)-t}d^{i}\binom{p-j}{i}b_{j}\\
	&=(-1)^{i}2^{2i-N}d^{i}\sum_{j=0}^{\fl{N}{2}-i}\binom{\fl{N}{2}-j}{i}b_{j}.
	\end{align*}
\end{proof}

\subsection{The proof of Theorem~\ref{D3bound}}\label{V111:appen}
When $\ell=-2$, for the same discussion as Theorem~\ref{D3bound}.
\begin{equation*}
\alpha_{6m}(14m-2)=-(\frac{14m-2}{3m})\sum_{i=0}^{6m-1}(-2)^{i}\binom{i+2m-2}{2m-2}\binom{12m-2-i}{6m-1}.
\end{equation*}
Let $n=2m$, and
\begin{equation*}
p_n=\sum_{i=0}^{3n-1}(-2)^{i}\binom{i+n-2}{n-2}\binom{6n-2-i}{3n-1}.
\end{equation*}
By using the Fast Zeilberger Package in Mathematica \cite{Fastzeil}, we obtain a recursion of $p_n$,
$$h_{n+2}p_{n+2}=h_{n+1}p_{n+1}+h_{n}p_{n},$$
where
\begin{align*}
h_{n+2}=&729 n (1 + n) (1 + 3 n) (2 + 3 n) (4 + 3 n) (5 + 3 n) (-135 - 371 n -
261 n^2 + 1247 n^3);\\
h_n=&448 (-2 + 7 n) (-1 + 7 n) (1 + 7 n) (2 + 7 n) (3 + 7 n) (4 +
7 n) (480 + 2848 n + 3480 n^2 + 1247 n^3);\\
h_{n+1}=&48 n (1 + 3 n) (2 + 3 n) (-775965 - 4093609 n - 9631385 n^2 -
9070501 n^3 - 1782359 n^4 + 1633193 n^5 + 496306 n^6).
\end{align*}
Note that $h_{n}$, $h_{n+1}$ and $h_{n+2}$ are all positive when $n\geq 3$, $p_1=6$, $p_2=120$, $p_3=3150$ and $p_4=80304$. This implies that $p_n>0$ for $n\geq 2$. Thus $\alpha_{6m}(14m-2)<0$ for $m\geq 1$.

 The same analysis holds for other cases. We can get the same recursion of $p_n$. So we only list the formula of $h_i$.

When $\ell=-3$,
\begin{align*}
h_{n+2}=&729 (-1 + n) (1 + n) (1 + 3 n) (2 + 3 n) (4 + 3 n) (5 + 3 n) (-504 -608 n - 1392 n^2 + 1247 n^3);\\
h_n=&448 (-3 + 7 n) (-2 + 7 n) (-1 + 7 n) (1 + 7 n) (2 + 7 n) (3 + 7 n) (-1257 + 349 n + 2349 n^2 + 1247 n^3);\\
h_{n+1}=&48 (1 + 3 n) (2 + 3 n) (3214890 + 9384249 n + 14090971 n^2 - 2171291 n^3 - 19263725 n^4 - 11599109 n^5\\
&+ 686749 n^6 + 496306 n^7).
\end{align*}

When $\ell=-4$,
\begin{align*}
h_{n+2}=&729 (-2 + n) (1 + n) (1 + 3 n) (2 + 3 n) (4 + 3 n) (5 + 3 n) (-1353 -563 n - 2523 n^2 + 1247 n^3);\\
h_{n}=&448 (-4 + 7 n) (-3 + 7 n) (-2 + 7 n) (-1 + 7 n) (1 + 7 n) (2 + 7 n) (-3192 - 1868 n + 1218 n^2 + 1247 n^3);\\
h_{n+1}=&48 (1 + 3 n) (2 + 3 n) (18830070 + 43348599 n + 74568645 n^2+40640779 n^3 - 13416795 n^4 - 24189239 n^5\\
&- 259695 n^6 + 496306 n^7).
\end{align*}

When $\ell=0$,
\begin{align*}
h_{n+2}=&729 (-2 + n) (1 + n) (2 + n) (2 + 3 n) (4 + 3 n) (5 + 3 n) (7 + 3 n) (-3645 - 3083 n - 783 n^2 + 1247 n^3);\\
h_{n}=&448 n (1 + 7 n) (2 + 7 n) (3 + 7 n) (4 + 7 n) (5 + 7 n) (6 + 7 n) (-6264 - 908 n + 2958 n^2 + 1247 n^3);\\
h_{n+1}=&48 (1 + n) (2 + 3 n) (4 + 3 n) (197026830 + 382545627 n +243647573 n^2 - 48990041 n^3 - 109302835 n^4\\
& - 33053867 n^5 +1425437 n^6 + 496306 n^7).
\end{align*}

When $\ell=1$,
\begin{align*}
h_{n+2}=&729 (-1 + n) (2 + n) (2 + 3 n) (4 + 3 n) (5 + 3 n) (7 + 3 n) (-1992 -1808 n + 348 n^2 + 1247 n^3);\\
h_{n}=&448 (1 + 7 n) (2 + 7 n) (3 + 7 n) (4 + 7 n) (5 + 7 n) (6 + 7 n) (-2205 + 2629 n + 4089 n^2 + 1247 n^3);\\
h_{n+1}=&48 (2 + 3 n) (4 + 3 n) (51994386 + 81931221 n - 12029301 n^2 -109375991 n^3 - 77973861 n^4 \\
&- 14829041 n^5 + 2371881 n^6 + 496306 n^7).
\end{align*}

When $\ell=2$,
\begin{align*}
h_{n+2}=&729 n (2 + n) (2 + 3 n) (4 + 3 n) (5 + 3 n) (7 + 3 n) (-795 - 251 n +1479 n^2 + 1247 n^3);\\
h_{n}=&448 (2 + 7 n) (3 + 7 n) (4 + 7 n) (5 + 7 n) (6 + 7 n) (8 +7 n) (1680 + 6448 n + 5220 n^2 + 1247 n^3);\\
h_{n+1}=&48 (2 + 3 n) (4 + 3 n) (-645120 - 28032777 n - 76821521 n^2 - 79862429 n^3 - 31964789 n^4\\
& + 622405 n^5 + 3318325 n^6 + 496306 n^7).
\end{align*}

When $\ell=3$,
\begin{align*}
h_{n+2}=&729 (1 + n) (2 + n) (2 + 3 n) (4 + 3 n) (5 + 3 n) (7 + 3 n) (192 +1588 n + 2610 n^2 + 1247 n^3);\\
h_{n}=&448 (3 + 7 n) (4 + 7 n) (5 + 7 n) (6 + 7 n) (8 +7 n) (9+7 n) (5637 + 10549 n + 6351 n^2 + 1247 n^3);\\
h_{n+1}=&48 (1 + n) (2 + 3 n) (4 + 3 n) (-5160960 - 12315591 n - 6140488 n^2 + 7365557 n^3 + 9532008 n^4 \\
&+ 3768463 n^5 + 496306 n^6).
\end{align*}

When $\ell=4$,
\begin{align*}
h_{n+2}=&729 (1 + n) (2 + n) (2 + 3 n) (4 + 3 n) (5 + 3 n) (7 + 3 n) (1215 + 2494 n + 1247 n^2);\\
h_{n}=&448 (4 + 7 n) (5 + 7 n) (6 + 7 n) (8 + 7 n) (9 + 7 n) (10 + 7 n) (4956 + 4988 n + 1247 n^2);\\
h_{n+1}=&48 (1 + n) (3 + 2 n) (2 + 3 n) (4 + 3 n) (1148175 + 3168228 n + 3289453 n^2 + 1488918 n^3 + 248153 n^4).
\end{align*}

When $\ell=5$,
\begin{align*}
h_{n+2}=&729 (-1 + n) (2 + n) (3 + n) (4 + 3 n) (5 + 3 n) (7 + 3 n) (8 + 3 n) (-4632 - 1568 n + 2088 n^2 + 1247 n^3);\\
h_{n}=&448 (1 + n) (5 + 7 n) (6 + 7 n) (8 + 7 n) (9 + 7 n) (10 + 7 n) (11 + 7 n) (-2865 + 6349 n + 5829 n^2 + 1247 n^3);\\
h_{n+1}=&48 (2 + n) (4 + 3 n) (5 + 3 n) (317472570 + 181921041 n -370788709 n^2 - 455647763 n^3 - 173370589 n^4\\
& - 15074909 n^5 + 4057013 n^6 + 496306 n^7).
\end{align*}

When $\ell=6$,
\begin{align*}
h_{n+2}=&729 n (2 + n) (3 + n) (4 + 3 n) (5 + 3 n) (7 + 3 n) (8 + 3 n) (-1743 + 1309 n + 3219 n^2 + 1247 n^3);\\
h_{n}=&448 (1 + n) (6 + 7 n) (8 + 7 n) (9 + 7 n) (10 + 7 n) (11 + 7 n) (12 + 7 n) (4032 + 11488 n + 6960 n^2 + 1247 n^3);\\
h_{n+1}=&48 (2 + n) (4 + 3 n) (5 + 3 n) (-23224320 - 267580053 n - 455857497 n^2 - 300018401 n^3 - 71856645 n^4\\
& + 6011233 n^5 + 5003457 n^6 + 496306 n^7).
\end{align*}

When $\ell=7$,
\begin{align*}
h_{n+2}=&729 (2 + n) (3 + n) (4 + 3 n) (5 + 3 n) (7 + 3 n) (8 + 3 n) (960 + 4468 n + 4350 n^2 + 1247 n^3)\\
h_{n}=&448 (8 + 7 n) (9 + 7 n) (10 + 7 n) (11 + 7 n) (12 + 7 n) (13 + 7 n) (11025 + 16909 n + 8091 n^2 + 1247 n^3);\\
h_{n+1}=&48 (2 + n) (4 + 3 n) (5 + 3 n) (-77414400 - 131658555 n -61786316 n^2 + 12280125 n^3 + 18870400 n^4\\
& + 5453595 n^5 +496306 n^6).
\end{align*}

When $\ell=8$,
\begin{align*}
h_{n+2}=&729 (2 + n) (3 + n) (4 + 3 n) (5 + 3 n) (7 + 3 n) (8 + 3 n) (3723 + 7909 n + 5481 n^2 + 1247 n^3);\\
h_{n}=&448 (8 + 7 n) (9 + 7 n) (10 + 7 n) (11 + 7 n) (12 + 7 n) (13 + 7 n) (18360 + 22612 n + 9222 n^2 + 1247 n^3);\\
h_{n+1}=&48 (2 + n) (4 + 3 n) (5 + 3 n) (11399265 + 53229759 n +86249173 n^2 + 67712283 n^3 + 28055911 n^4\\
&+ 5903733 n^5 +496306 n^6).
\end{align*}

When $\ell=9$,
\begin{align*}
h_{n+2}=&729 (2 + n) (3 + n) (4 + 3 n) (5 + 3 n) (7 + 3 n) (8 + 3 n) (3396 + 4118 n + 1247 n^2);\\
h_{n}=&448 (9 + 7 n) (10 + 7 n) (11 + 7 n) (12 + 7 n) (13 + 7 n) (15 + 7 n) (8761 + 6612 n + 1247 n^2);\\
h_{n+1}=&48 (2 + n) (4 + 3 n) (5 + 3 n) (19995525 + 49463220 n + 47943485 n^2 + 22845248 n^3 + 5361259 n^4 + 496306 n^5).
\end{align*}

\subsection{The proof of Proposition~\ref{thm:AME}}\label{appendix:d4}
(i)	 For an AME state in $\bbC^{4}\otimes(\bbC^{2})^{\otimes 2n}$, by Eq.~\eqref{eq:shadow}, we obtain that $s_1=2f(n)$, where
\begin{equation}
f(n)=(2n+1)+\sum_{i=1}^{n}(-1)^{i}(2n+1-2i)\left[\binom{2n}{i}\frac{1}{2^i}+\binom{2n}{i-1}\frac{1}{4\times 2^{i-1}}\right],
\end{equation}
and $s_3=2g(n)$, where
\begin{equation}
\begin{aligned}
g(n)=&\binom{2n+1}{3}+\sum_{i=1}^{n}\left[(-1)^{i-3}\binom{i}{3}+(-1)^{i-2}(2n+1-i)\binom{i}{2}+(-1)^{i-1}\binom{2n+1-i}{2}i+(-1)^i\binom{2n+1-i}{3}\right]\\
&\left[\binom{2n}{i}\frac{1}{2^i}+\binom{2n}{i-1}\frac{1}{4\times 2^{i-1}}\right].
\end{aligned}
\end{equation}	

There are two cases. If $n=2m$, by using the Fast Zeilberger Package in Mathematica \cite{Fastzeil}, we have
\begin{equation}
f(2m+2)=p_m\cdot f(2m)-q_m,
\end{equation}
where
\begin{equation}
\begin{aligned}
p_m&=\frac{41+36m}{16 (5 + 36 m)},\\
q_m&=\frac{(2 + m) (5 + 2 m) (7 + 97 m + 435 m^2 + 252 m^3)\binom{7 + 4 m}{2 + 2 m}\times 3\times2^{-1 - 2 m} }{
	16 (5 + 36 m) (1 + 4 m) (3 + 4 m) (5 + 4 m) (7 + 4 m)}.
\end{aligned}
\end{equation}
Since  $p_m>0$, $q_m>0$ when $m\geq 2$, and $f(4)=-\frac{7}{8}<0$, we have $f(2m+2)<0$ for $m\geq 2$.
Thus
\begin{equation}\label{eq:d4S1}
s_1=2f(n)<0 \quad \text{for} \quad n=2m\geq 4.
\end{equation}

If $n=2m+1$, by using the Fast Zeilberger Package in Mathematica \cite{Fastzeil}, we have
\begin{equation}
g(2m+3)=u_m\cdot g(2m+1)-v_m,
\end{equation}
where
\begin{equation}
\begin{aligned}
u_m&=\frac{(3 + 2 m) (5 + 4 m) (17 + 12 m)}{16 (1 + 2 m) (1 + 4 m) (5 + 12 m)},\\
v_m&=\frac{ (3 + m) (1 + 2 m) (3 + 2 m) (5 + 2 m) (39 + 916 m + 2452 m^2 +
	1008 m^3) \binom{9+4m}{3 + 2 m}\times 2^{-4 - 2 m} }{
	16 (1 + 2 m) (1 + 4 m) (5 + 12 m) (3 + 4 m) (7 + 4 m) (9 + 4 m)}.
\end{aligned}
\end{equation}

Since  $u_m>0$, $v_m>0$ when $m\geq 2$, and $g(5)=-\frac{225}{32}<0$, we have $g(2m+3)<0$ for $m\geq 2$.
Thus
\begin{equation}\label{eq:d4S3}
s_3=2g(n)<0 \quad \text{for} \quad n=2m+1\geq 5.
\end{equation}
By Eqs.~\eqref{eq:shadow}, \eqref{eq:d4S1} and \eqref{eq:d4S3}, we know that
AME states in $\bbC^4\otimes(\bbC^{2})^{\otimes 2n}$ do not exist when $n\geq 4$.

(ii) For an AME state in $\bbC^{d}\otimes(\bbC^{3})^{\otimes 2n}$, by Eq.~\eqref{eq:shadow},
 $s_1=2f(n)$, where
	\begin{equation}
	f(n)=(2n+1)+\sum_{i=1}^{n}(-1)^{i}(2n+1-2i)\left[\binom{2n}{i}\frac{1}{3^i}+\binom{2n}{i-1}\frac{1}{d\times 3^{i-1}}\right],
	\end{equation}
	and $s_3=2g(n)$, where
	\begin{equation}
	\begin{aligned}
	g(n)=&\binom{2n+1}{3}+\sum_{i=1}^{n}\left[(-1)^{i-3}\binom{i}{3}+(-1)^{i-2}(2n+1-i)\binom{i}{2}+(-1)^{i-1}\binom{2n+1-i}{2}i+(-1)^i\binom{2n+1-i}{3}\right]\\
	&\left[\binom{2n}{i}\frac{1}{3^i}+\binom{2n}{i-1}\frac{1}{d\times 3^{i-1}}\right].
	\end{aligned}
	\end{equation}
There are two cases. If $n=2m$, then
\begin{equation}\label{eq:fm}
  f(2m+2)=p_mf(2m)-q_m.
\end{equation}
 If $n=2m+1$, then
 \begin{equation}\label{eq:gm}
  g(2m+3)=u_mg(2m+1)-v_m.
\end{equation}
Similarly, for we use the Fast Zeilberger Package in Mathematica \cite{Fastzeil} to calculate $p_m$, $q_m$, $u_m$ and $v_m$ for $d=2$ and $4\leq d\leq 9$.

When $d=2   $,
\begin{equation*}
\begin{aligned}
    p_m=&\frac{16 (11+8 m)}{81 (3+8 m)},\\
    q_m=&\frac{3^{-2m}(2+m) (5+2 m) (-53+388 m+1856 m^2+1280 m^3)\binom{7+4m}{2+2m}}{324 (1+4 m) (3+4 m) (5+4 m) (7+4 m)(3+8 m)},\\
    u_m=&\frac{16 (3+2 m) (5+4 m) (17+8 m)}{81 (1+2 m) (1+4 m) (9+8 m)},\\
    v_m=&\frac{3^{-1-2 m}(3+m) (1+2 m) (3+2 m) (5+2 m) (-771+1220 m+3136 m^2+1280 m^3)\binom{7+2 (1+2 m)}{3+2 m}}{324 (3+4 m) (7+4 m) (9+4 m)(1+2 m) (1+4 m) (9+8 m) }.
\end{aligned}
\end{equation*}
Note that $f(6)<0$, and $g(9)<0$. By Eqs.~\eqref{eq:fm}
and \eqref{eq:gm}, we know  that
$f(2m+2)<0$ for $m\geq 2$, and $g(2m+3)<0$ for $m\geq 3$.

When $d=4$,
\begin{equation*}
\begin{aligned}
    p_m=&\frac{16 (29+24m)}{81 (5+24 m)},\\
    q_m=&\frac{3^{-2m} (2+m) (5+2 m) (-169+276 m+7360 m^2+5376 m^3)\binom{7+4m}{2+2m}}{648 (1+4 m) (3+4 m) (5+4 m) (7+4 m) (5+24 m)},\\
    u_m=&\frac{16 (3+2 m) (5+4 m) (13+8 m)}{81 (1+2 m) (1+4 m) (5+8 m) },\\
    v_m=&\frac{3^{-1-2 m}(3+m) (1+2 m) (3+2 m) (5+2 m) (-789+284 m+3648 m^2+1792 m^3)\binom{7+2 (1+2 m)}{3+2 m}}{648 (3+4 m) (7+4 m) (9+4 m) (1+2 m) (1+4 m) (5+8 m) }.
\end{aligned}
\end{equation*}
Note that $f(6)<0$, and $g(9)<0$. By Eqs.~\eqref{eq:fm}
and \eqref{eq:gm}, we know  that
$f(2m+2)<0$ for $m\geq 2$, and $g(2m+3)<0$ for $m\geq 3$.

When $d=5$,
\begin{equation*}
\begin{aligned}
    p_m=&\frac{16 (19+16 m) }{81 (3+16 m)},\\
    q_m=&\frac{2\times 3^{-2 m} (2+m) (5+2 m) (-1+8 m) (25+192 m+128 m^2)\binom{7+4 m}{2+2 m}}{405 (1+4 m) (3+4 m) (5+4 m) (7+4 m) (3+16 m) },\\
    u_m=&\frac{16 (3+2 m) (5+4 m) (25+16 m)}{81 (1+2 m) (1+4 m) (9+16 m) },\\
    v_m=&\frac{2\times 3^{-1-2 m} (3+m) (1+2 m) (3+2 m) (5+2 m) (-417+40 m+2048 m^2+1024 m^3) \binom{7+2 (1+2 m)}{3+2 m}}{405 (3+4 m) (7+4 m) (9+4 m) (1+2 m) (1+4 m) (9+16 m) }.
\end{aligned}
\end{equation*}
Note that $f(6)<0$, and $g(9)<0$. By Eqs.~\eqref{eq:fm}
and \eqref{eq:gm}, we know  that
$f(2m+2)<0$ for $m\geq 2$, and $g(2m+3)<0$ for $m\geq 3$.

When $d=6$,
\begin{equation*}
\begin{aligned}
    p_m=&\frac{16 (47+40 m) }{81 (7+40 m) },\\
    q_m=&\frac{3^{-2 m} (2+m) (5+2 m) (-71-84 m+5312 m^2+3840 m^3) \binom{7+4 m}{2+2 m}}{324(1+4 m) (3+4 m) (5+4 m) (7+4 m)(7+40 m) },\\
    u_m=&\frac{16 (3+2 m) (5+4 m) (61+40 m) }{81 (1+2 m) (1+4 m) (21+40 m) },\\
    v_m=&\frac{3^{-1-2 m} (3+m) (1+2 m) (3+2 m) (5+2 m) (-1473-148 m+7616 m^2+3840 m^3)\binom{7+2 (1+2 m)}{3+2 m}}{324 (3+4 m) (7+4 m) (9+4 m)(1+2 m) (1+4 m)(21+40 m)  }.
\end{aligned}
\end{equation*}
Note that $f(6)<0$, and $g(9)<0$. By Eqs.~\eqref{eq:fm}
and \eqref{eq:gm}, we know  that
$f(2m+2)<0$ for $m\geq 2$, and $g(2m+3)<0$ for $m\geq 3$.

When $d=7$,
\begin{equation*}
\begin{aligned}
    p_m=&\frac{16 (7+6 m)  }{81 (1+6 m)  },\\
    q_m=&\frac{3^{-2 m} (2+m) (5+2 m) (-13-42 m+1336 m^2+960 m^3)\binom{7+4 m}{2+2 m}}{567(1+4 m) (3+4 m) (5+4 m) (7+4 m)(1+6 m) },\\
    u_m=&\frac{16 (3+2 m)^2 (5+4 m) }{81 (1+2 m)^2 (1+4 m) },\\
    v_m=&\frac{3^{-1-2 m} (3+m) (1+2 m) (3+2 m) (5+2 m) (-117-30 m+632 m^2+320 m^3)\binom{7+2 (1+2 m)}{3+2 m}}{567 (3+4 m) (7+4 m) (9+4 m)(1+2 m)^2 (1+4 m)  }.
\end{aligned}
\end{equation*}
Note that $f(6)<0$, and $g(9)<0$. By Eqs.~\eqref{eq:fm}
and \eqref{eq:gm}, we know  that
$f(2m+2)<0$ for $m\geq 2$, and $g(2m+3)<0$ for $m\geq 3$.

When $d=7$,
\begin{equation*}
\begin{aligned}
    p_m=&\frac{16 (7+6 m)  }{81 (1+6 m)  },\\
    q_m=&\frac{3^{-2 m} (2+m) (5+2 m) (-13-42 m+1336 m^2+960 m^3)\binom{7+4 m}{2+2 m}}{567(1+4 m) (3+4 m) (5+4 m) (7+4 m)(1+6 m) },\\
    u_m=&\frac{16 (3+2 m)^2 (5+4 m) }{81 (1+2 m)^2 (1+4 m) },\\
    v_m=&\frac{3^{-1-2 m} (3+m) (1+2 m) (3+2 m) (5+2 m) (-117-30 m+632 m^2+320 m^3)\binom{7+2 (1+2 m)}{3+2 m}}{567 (3+4 m) (7+4 m) (9+4 m)(1+2 m)^2 (1+4 m)  }.
\end{aligned}
\end{equation*}
Note that $f(6)<0$, and $g(9)<0$. By Eqs.~\eqref{eq:fm}
and \eqref{eq:gm}, we know  that
$f(2m+2)<0$ for $m\geq 2$, and $g(2m+3)<0$ for $m\geq 3$.

When $d=8$,
\begin{equation*}
\begin{aligned}
    p_m=&\frac{16 (65+56 m)   }{81 (9+56 m)  },\\
    q_m=&\frac{3^{-2 m} (2+m) (5+2 m) (-185-1196 m+27584 m^2+19712 m^3)\binom{7+4 m}{2+2 m}}{1296(1+4 m) (3+4 m) (5+4 m) (7+4 m) (9+56 m) },\\
    u_m=&\frac{16 (3+2 m) (5+4 m) (83+56 m)  }{81 (1+2 m) (1+4 m) (27+56 m) },\\
    v_m=&\frac{3^{-1-2 m} (3+m) (1+2 m) (3+2 m) (5+2 m) (-6927-2668 m+38848 m^2+19712 m^3) \binom{7+2 (1+2 m)}{3+2 m}}{1296(3+4 m) (7+4 m) (9+4 m)(1+2 m) (1+4 m) (27+56 m) }.
\end{aligned}
\end{equation*}
Note that $f(6)<0$, and $g(9)<0$. By Eqs.~\eqref{eq:fm}
and \eqref{eq:gm}, we know  that
$f(2m+2)<0$ for $m\geq 2$, and $g(2m+3)<0$ for $m\geq 3$.

When $d=9$,
\begin{equation*}
\begin{aligned}
    p_m=&\frac{16 (37+32 m)   }{81 (5+32 m)  },\\
    q_m=&\frac{4\times 3^{-1-2 m} (2+m) (5+2 m) (-3-38 m+720 m^2+512 m^3) \binom{7+4 m}{2+2 m}}{81(1+4 m) (3+4 m) (5+4 m) (7+4 m) (5+32 m) },\\
    u_m=&\frac{16 (3+2 m) (5+4 m) (47+32 m)  }{81 (1+2 m) (1+4 m) (15+32 m)  },\\
    v_m=&\frac{8\times 3^{-2-2 m} (3+m) (1+2 m) (3+2 m) (5+2 m) (-87-43 m+504 m^2+256 m^3) \binom{7+2 (1+2 m)}{3+2 m}}{81(3+4 m) (7+4 m) (9+4 m)(1+2 m) (1+4 m) (15+32 m)  }.
\end{aligned}
\end{equation*}
Note that $f(6)<0$, and $g(11)<0$. By Eqs.~\eqref{eq:fm}
and \eqref{eq:gm}, we know  that
$f(2m+2)<0$ for $m\geq 2$, and $g(2m+3)<0$ for $m\geq 4$.

\vskip 10pt
\bibliographystyle{IEEEtran}
\bibliography{reference}


\end{document}